\documentclass[11pt]{article}
\usepackage{amssymb}
\usepackage{amsmath}
\usepackage{color}

\newenvironment{proof}{\noindent {\bf Proof }}
{\hfill $\bullet$ \vspace{0.25cm}}

\def\La {{\Lambda}}
\def\si {{\sigma}}
\def\Si {{\Sigma}}
\def\la {{\lambda}}
\def\ga{{ \gamma}}
\def\Ga{{ \Gamma}}

\def\Om{{ \Omega}}
\newcommand{\und}{\underline}
\newcommand{\nn}{\nonumber}
\newcommand{\dis}{\displaystyle}

\newtheorem{thm}{Theorem}
\newtheorem{prop}{\indent Proposition}

\newtheorem{lem}{\indent Lemma}

\newtheorem{cor}{\indent Corollary}

\newcommand{\mmmintone}[1]{{\dis{\int\kern -.36cm-}}_{\kern-.21cm\substack{#1}}\;\;}
\newcommand{\mmmintwo}[2]{{\dis{\int\kern -.43cm-}}_{\kern-.21cm\substack{#1}}^{\substack{#2}}\;\;}
\newcommand{\submint}{{\scriptstyle{\int\kern -.66em -}}}
\newcommand{\submintone}[1]{{\scriptstyle{\int\kern -.66em-}}_{\scriptscriptstyle{\kern-.21em\substack{#1}}}}
\newcommand{\fracmint}{{\textstyle{\int\kern -.88em -}}}
\newcommand{\fracmintone}[1]{{\textstyle{\int\kern -.88em
-}}_{\scriptscriptstyle{\kern-.21em\substack{#1}}}\;}

\title{Scaling limits in highly anisotropic systems}

\title{Highly anisotropic scaling limits}

\author{M. Cassandro \footnote{ E-mail: marzio.cassandro@gmail.com  }, M. Colangeli \footnote{ E-mail: matteo.colangeli@gssi.infn.it }, E. Presutti\footnote{  E-mail: errico.presutti@gmail.com } \\Gran Sasso Science Institute, Via. F. Crispi 7, 00167 L' Aquila, Italy}

\date{\today}

\begin{document}

\maketitle

\begin{abstract}
We consider a highly anisotropic $d=2$ Ising spin model whose precise definition can be found at the beginning of Section \ref{Fsec.2}.
In this model the
spins on a same horizontal line (layer) interact via a $d=1$ Kac potential
while the vertical interaction is between nearest neighbors, both interactions being ferromagnetic. The temperature is set equal to 1 which is the mean field critical value, so that the mean field limit for
the Kac potential alone does not have
a spontaneous magnetization.  We compute
the phase diagram of the full system in the Lebowitz-Penrose limit showing that due to the vertical interaction it has a spontaneous magnetization.  The result  is not covered by the  Lebowitz-Penrose theory
because our  Kac potential has
support on regions of positive codimension.

\end{abstract}
%
%

%
%
%



\section{Introduction}
\label{Fsec.1}

This work focuses on the proof of the Lebowitz-Penrose limit for a highly anisotropic $d=2$ Ising spin model which has been first studied in \cite{FMMPV2}, its precise definition can be found at the beginning of Section \ref{Fsec.2}.
In this model the
spins on a same horizontal line (layer) interact via a $d=1$ Kac potential
while the vertical interaction is between nearest neighbors, both interactions are ferromagnetic. The temperature is set equal to 1 which is the mean field critical value (without vertical interactions), so that the mean field limit for
the Kac potential alone does not have spontaneous magnetization. However in \cite{FMMPV2}
it is proved that even a small vertical interaction is sufficient to produce a
phase transition at least for small values of the Kac scaling parameter $\ga$. The idea in \cite{FMMPV2}  is to study a model with fewer vertical interactions (those left have a chessboard structure):  by  the Ginibre inequalities if  a spontaneous magnetization
is present in the reduced model then it is also present in the true system as well.
The advantage of working in the reduced system is that one can split the system into blocks of two layers, the vertical interaction is left only inside each block so that blocks do not interact vertically with each other;  the horizontal interaction is unchanged. As a consequence in  \cite{FMMPV2} it is shown that it is sufficient
to carry out the Lebowitz-Penrose coarse graining procedure only for two-layer systems. It is then
proved that this can be done, that there is a positive spontaneous magnetization in the limits volume to infinity and then $\ga\to 0$ and that such a property remains valid also at finite small $\ga> 0$.  However the value of the spontaneous magnetization for the reduced system  is  certainly smaller than the real one because half of the vertical interactions has been dropped.

The problem of studying directly the original system and in particular to
find its true spontaneous magnetization
has been left open in \cite{FMMPV2}, we attack it in this paper
determining explicitly the limit phase diagram of the
true system when first the volume goes to infinity and then $\ga \to 0$.
This is not covered by the  Lebowitz-Penrose theory because our  Kac potential is singular having
support on regions of positive codimension.   We hope in a successive paper to prove that there is a positive spontaneous magnetization also at $\ga>0$ which converges as  $\ga\to 0$ to the one found here.

This work is part of a more general project (which besides us involves several other colleagues) where we
want to study  systems with Kac potentials having support on regions of  positive codimension
plus short range interactions, both in equilibrium and non equilibrium.
The description of the system is hybrid: referring to our Ising model
we can make a coarse graining on each layer and introduce macroscopic variables
but the interaction between layers is microscopic and it is described by
an effective interaction to be determined.  The purpose is to derive such an effective hamiltonian and find its ground states. In this paper we compute the limit
ground state energy but we hope in the future
to study the excited states and derive the large deviations functional.

Similar structures are present in SOS models, for instance in the SOS interface models where the real valued spin variables $S_x$, $x\in \mathbb Z^d$, represent the position of an interface in $\mathbb Z^{d+1}$.  Evidently the model is obtained by an anisotropic scaling limit for which the interface becomes sharp,the point  $S_x$  on the ``vertical'' line through $x$, while the interaction among spins remains short range. We hope to establish such connections starting from models like the one considered here.

The appearance of a
macroscopic description on the layers
may also originate from a canonical constraint with or without the presence of a Kac potential.
Considering the system in a finite box we may fix on each horizontal layer the total magnetization which gives rise to a multi-canonical ensemble.  Indeed when we study
the system with Kac potentials following the Lebowitz-Penrose procedure we coarse grain and get such multi-canonical ensembles.  Our analysis will be based on a proof that equivalence of ensembles extends to such cases.

The multi-canonical constraint appears naturally in dynamical problems when we consider a Kawasaki dynamics on each layer so that the total magnetization on each layer
is constant in time. The vertical interaction affects the rates of horizontal
exchanges on the layers so that in the hydrodynamic limit
the evolution is  conjectured to be ruled by coupled diffusions.
An interesting variant would be a weakly asymmetric simple exclusion on each layer
with small interactions among layers which should be in the KPZ class of systems.
We refer to the introductions in  \cite{FMMPV1}-- \cite{FMMPV2} for more references and
a list of open problems and conjectures, in particular the connection with quantum Ising models (via Feynman-Kac), phase transitions for the hard-rods Kac-Helfand model and the dependence
on $\ga$ of the critical value  of the vertical interaction  for a phase transition to occur.

We conclude the introduction by observing that
highly anisotropic interactions are present in nature, the best
example is the graphite where horizontal structures are rather
free to slide one with respect to the other.  However it may happen that
even a small interaction among layers produces macroscopic effects.  For instance
for bilayer graphene samples interacting via an interlayer coupling constant \cite{Zhang,Rutt,LeRoy}  the presence of a band gap in the energy spectrum, which is tunable by an external electric field, paves the way to a variety of applications in electronics \cite{Schwierz}.\\
Multilayer graphene samples have also gone, recently, under intense investigation \cite{Shahil,Yanko}, which revealed the rise of exceptional thermal conduction properties for these materials as well as the possibility to control the thermodynamically stable cristalline structure of the material through an external voltage.

\vskip1cm

\setcounter{equation}{0}

\section{The model and the main result} 
\label{Fsec.2}

As mentioned in the introduction one of our aims
is
the extension of the Lebowitz-Penrose theory
to   cases where the support of the Kac potential has a
positive codimension.
This is what we accomplish in this paper
in the simple context of
the $d=2$  Ising model.  Let $\La$ be a square
in
$\mathbb Z^2$, $L$ its side, $(x,i)$ its points.  Write $\si \in \{-1,1\}^\La$ for a  spin configuration  in $\La$, define  $\si(x,L+1)=\si(x,1)$, $\si(x+L,i)=\si(x,i)$ and let
 \begin{eqnarray}
    \label{102.1}
H^{\rm per}_{\ga, h_{\rm ext},L}(\si) &=& H_{\ga,L}(\si) + H^{\rm vert}_L(\si) + H_{h_{\rm ext},L}(\si)
 \\
 H_{\ga,L}(\si) &=&   \sum_{i=1}^L \{-\frac 12 \sum_{x\ne y}J_\ga(x,y)\si(x,i)\si(y,i)\}\nn\\
 H^{\rm vert}_L(\si) &=&   \sum_{ x=1}^L\{ - \la \sum_{i=1}^L \si(x,i)\si(x,i+1)\}\nn\\
H_{h_{\rm ext},L}(\si) &=&
 -\sum_{(x,i)\in \La}
h_{\rm ext}\si(x,i)\nn
   \end{eqnarray}
$H_{\ga,L}(\si)$ is the Kac hamiltonian, it has only horizontal interactions; $H^{\rm vert}_L(\si) $ is the hamiltonian of a nearest neighbor Ising model with only vertical interactions; $H_{h_{\rm ext},L}(\si)$ is the energy due to the external magnetic field $h_{\rm ext}$.  We suppose that
 \begin{equation}
    \label{102.1.1}
J_\ga(x,y)= c_\ga \ga J(\ga x,\ga y)
   \end{equation}
where $J(r,r')$ is a smooth, symmetric probability kernel on $\mathbb R$ which vanishes
for $|r-r'|\ge 1$; $c_\ga$ is such that
 \begin{equation}
    \label{102.1.2}
\sum_y J_\ga(x,y)= 1
   \end{equation}
Since $\int J(r,r')dr'=1$, $c_\ga \to 1 $ as $\ga\to 0$.
Let
 \begin{equation}
    \label{102.2}
Z_{\ga, h_{\rm ext},L}^{\rm per}= \sum_{\si\in \{-1,1\}^\La}e^{- H^{\rm per}_{\ga, h_{\rm ext},L}(\si)}
   \end{equation}
be the   partition function relative to the hamiltonian $H^{\rm per}_{\ga, h_{\rm ext},L}(\si)$. Call $f_\la(m)$  the free energy density with magnetization density $m$ relative to the hamiltonian $H^{\rm vert}_L(\si)$, since the horizontal interactions are absent $f_\la(m)$
is equal to the free energy of the $d=1$ Ising model with only nearest neighbor interactions of strength $\la$.

\medskip

\begin{thm}
\label{thm101.1}
For $\la\ge 0$
small enough
 \begin{equation}
    \label{102.3}
\lim_{\ga\to 0}\lim_{L\to \infty}\frac{\log Z_{\ga, h_{\rm ext},L}^{\rm per}}{|\La|}= -
\inf_{m\in [-1,1]} \Big\{ -h_{\rm ext} m + [ - \frac{m^2}{2}  + f_\la (m)]\Big\}
   \end{equation}

\end{thm}

\medskip

After a few comments on Theorem \ref{thm101.1}
we give a heuristic derivation of \eqref{102.3} followed by
a description of how proofs are organized in the various sections.

\medskip
\subsection{ Remarks on Theorem \ref{thm101.1}.}

\begin{itemize}

\item   \eqref{102.3} remains valid for general Van Hove regions and boundary conditions
since the interaction has finite range for any fixed value of  $\ga> 0$.
The restriction to small $\la$ is needed for cluster expansion, it is
technical and could be presumably removed.

\item  The limit in \eqref{102.3} is the sum of the external magnetic field energy
$ -h_{\rm ext} m$, the mean field energy $-m^2/2$ and the vertical free energy $f_\la(m)$: it reflects the analogous splitting of the hamiltonian in \eqref{102.1}.

\item $\lim_{L\to \infty}\frac{\log Z_{\ga, h_{\rm ext},L}^{\rm per}}{|\La|}=:P_\ga(h_{\rm ext})$ is the pressure of the system with hamiltonian $H^{\rm per}_{\ga, h_{\rm ext},L}$.  By ferromagnetic inequalities $P_\ga(h_{\rm ext})$ is for any $\ga>0$ a convex function of $h_{\rm ext}$  differentiable at any $h_{\rm ext}\ne 0$; its derivative is the magnetization which is equal to the average spin for the unique DLR measure at the given values of $h_{\rm ext}$ and $\ga$.  The  limits (by subsequences)
    of $P_\ga(h_{\rm ext})$ as $\ga\to 0$ are thus convex functions and Theorem \ref{thm101.1} proves that the limit actually exists (without going to subsequences) and identifies its value.

\item  The limit of $P_\ga(h_{\rm ext})$ as $\ga\to 0$  is the pressure
$P(h_{\rm ext})$ in the Lebowitz-Penrose limit  when first $|\La| \to \infty$
and then $\ga\to 0$. \eqref{102.3} shows that $P(h_{\rm ext})$ is the Legendre transform of
the function $[ - \frac{m^2}{2}  + f_\la (m)]$ and therefore the free energy $F_\la(m)$
defined as the Legendre transform of the  pressure $P(h_{\rm ext})$ is equal to the convex envelope:
 \begin{equation}
    \label{101.4}
F_\la(m)= CE  \Big\{ - \frac{m^2}{2}  + f_\la (m)\Big\}
   \end{equation}

\item
\eqref{101.4} is in agreement with the Lebowitz-Penrose result which states that the limit free energy density is the convex envelope of  $- \frac{m^2}{2}$ plus the free energy density of the reference system (i.e.\ without the Kac interaction).  The  Lebowitz-Penrose analysis however applies if the Kac interaction is non degenerate being
positive in two dimensional regions.  Our theorem shows that this is not necessary.

\item When $\la=0$,  $f_0 (m) = -S(m)$ where $S(m)$ is the entropy of the
free Ising model:
 \begin{equation}
    \label{101.4.00}
-S(m) = \frac{1+m}{2} \log \frac{1+m}{2} +  \frac{1-m}{2} \log \frac{1-m}{2}
\end{equation}
In this case $ - \frac{m^2}{2}  + f_0 (m)$ is strictly convex and coincides with $F_0(m)$.
When $\la>0$ we shall see that the function  $ - \frac{m^2}{2}  + f_\la (m)$ is no longer convex.  In fact the Taylor expansion of $-S(m)$ gives
  \begin{equation}
    \label{101.4.00.01}
 -S(m) = -\log 2 +\sum_{k=0}^\infty \frac 1{2k+1}\frac 1{2k+2} m^{2k+2}
 \end{equation}
and to leading orders in $\la$, $f_\la(m) = -S(m) - \la m^2$
so that
$\dis{ - \frac{m^2}{2}  + f_\la (m)}$
has a  double well shape with minima at $\pm \sqrt{6\la}$ and
$F_\la(m)$ is constant in the interval with endpoints $\pm \sqrt{6\la}$. The spontaneous magnetization is then  $\sqrt{6\la}$ to be compared with the value $\sqrt{3\la}$ found  in  \cite{FMMPV2} for the system with reduced vertical interactions, as described in the introduction.


\item The proof of Theorem \ref{thm101.1} does not require the use of
a non local free energy functional as the one introduced by Lebowitz-Penrose, but we have
nonetheless established
some basic ingredients for its derivation which will be used in
a future work to study
the large deviations.

\end{itemize}

\vskip.5cm

\subsection {Heuristic derivation of the mean field equation}
\label{Fsubsec.2.2}
Let $\langle \si(x,i)\rangle=:m$ be the average spin in an extremal, translation invariant DLR measure
at $\ga>0$.  Then
 \begin{equation}
    \label{101.5}
\langle \si(x,i)\rangle =\langle \tanh\{ \sum_y J_\ga(x,y) \si(y,i) +
\la[\si(x,i+1)+\si(x,i-1)] + h_{\rm ext}\}\rangle
   \end{equation}
By the law of large numbers $\sum_y J_\ga(x,y) (\si(y,i)-m) \to 0$
in the limit $\ga\to 0$,
recall that $\sum_y J_\ga(x,y)=1$.
In such an approximation \eqref{101.5} becomes
 \begin{equation}
    \label{101.6}
\langle \si(x,i)\rangle =\langle \tanh\{
\la[\si(x,i+1)+\si(x,i-1)] + h_{\rm ext}+m \}\rangle
   \end{equation}
This is the equation for the average spin in a $d=1$
Ising model with only nearest neighbor interactions of strength $\la$
and magnetic field $h_{\rm ext}+m$.  Then the average
spin is equal to the thermodynamic magnetization $m$ which is related to the free energy $f_\la(m)$ by a variational principle which gives
 \begin{equation}
    \label{101.7}
0= (h_{\rm ext}+m) - f'_\la(m)= h_{\rm ext} - \frac d{dm}\Big(-\frac{m^2}2+f_\la(m)\Big)
   \end{equation}
in agreement with   \eqref{102.3}--\eqref{101.4}.

\vskip.5cm

\subsection {Organization of the paper}

The proof of \eqref{102.3}
distinguishes
large and small values of the magnetization and consequently of the magnetic field.
Large  magnetic fields are studied in Section \ref{Fsec.3} by using the Dobrushin high temperature theory based on the
Vaserstein distance; the ``small'' values of the magnetic field are studied in the remaining sections. In Section \ref{Fsec.4} we give the scheme of   proof of Theorem \ref{thm101.1} which is based on
the following steps (each step being discussed in a subsection). (1) a coarse graining procedure a la Lebowitz-Penrose to reduce to a $d=1$ system with only nearest neighbor interactions and without Kac potentials.  The price is that we have a variational problem with multiple constraints as
we have fixed the magnetization
on each layer. (2)  We then consider the analogous problem in the multi gran canonical ensemble
where on each layer we have a magnetic field.  The partition function of such a system is studied in details using cluster expansion under the assumption that $\la$ is sufficiently small.
(3)  We prove an extended equivalence of ensembles so that the original variational problem with
constraints given by the magnetization is replaced by a variational problem where one needs to optimize on the value of the auxiliary magnetic fields.
(4) The proof proceeds by showing that the minimizer is made by magnetic fields equal to each other on each layer. (5) We then show  that Theorem \ref{thm101.1} follows.

In Section \ref{Fsec.5} we prove a combinatorial lemma which says that
any monomial $u_1^{n_1}\cdots u_k^{n_k}$ in the variables $u_1,..,u_k$, $n_1+\cdots n_k=N\ge 2$, can be written as a sum of one body monomials $p_iu_i^N$, $p_i$ positive numbers, plus a sum of terms proportional to gradients squared, $\sum_{i<j} d_{i,j} (u_i-u_j)^2$, the $d_{i,j}$ polynomials of degree $N-2$.
This is the essential property needed to prove that the minimizers are homogeneous.

The proofs of all the above statements are   reported in successive appendices.

\vskip1cm

\setcounter{equation}{0}

\section{Large magnetic fields}
\label{Fsec.3}
%

The   heuristic argument presented in Subsection \ref{Fsubsec.2.2} is made rigorous for large magnetic fields in the following  theorem.

\medskip

\begin{thm}
\label{thm101.2}
For any $\la>0$ let $h_{\rm ext}>0$ be so large that
 \begin{equation}
    \label{101.8}
r:=\frac{1+2\la} {\cosh^2( h_{\rm ext}-1-2\la)} < \frac 14
   \end{equation}
Then (i) for any $\ga>0$ there is a unique DLR measure (by ferromagnetic inequalities
the statement
actually holds for any $h_{\rm ext}\ne 0$); (ii)  its
magnetization $m_\ga$ (the average value of a spin) converges as $\ga\to 0$ to the value $m$ for which \eqref{101.7} holds; (iii)  $m$ is the unique minimizer of \eqref{102.3} and
 \begin{equation}
    \label{102.3.00}
\lim_{\ga\to 0}\lim_{L\to \infty}\frac{\log Z_{\ga, h_{\rm ext},L}^{\rm per}}{|\La|}= -
\Big\{ -h_{\rm ext} m + [ - \frac{m^2}{2}  + f_\la (m)]\Big\}
   \end{equation}

\end{thm}

\medskip

As we have already mentioned the proof of Theorem
\ref{thm101.2}  is based on  the techniques introduced
by Dobrushin to prove his famous large temperature uniqueness theorem.
In this way we get uniqueness of the DLR measures and
exponential decay of correlations  for any fixed $\ga>0$.
We then use an interpolation procedure to derive the
phase diagram of the system which was introduced in
\cite{merola} to study the corrections
in $\ga$ to the mean field limit and thus prove Theorem
\ref{thm101.2}.  The details are reported in Appendix \ref{app.A}.

\vskip1cm

\setcounter{equation}{0}

\section{Theorem \ref{thm101.1}: scheme of   proof}   
\label{Fsec.4}

Theorem
\ref{thm101.1} is thus proved for large magnetic fields
and the remaining part of the paper deals
with the ``bounded'' magnetic fields.
To be precise we suppose hereafter $\la \in(0,1)$, but
further requests on the smallness of $\la$ will be asked later on, and restrict to
magnetic fields
 \begin{equation}
    \label{101.17}
h_{\rm ext}\in [0,h^*],\quad  h^*:=\frac{3} {\cosh^2( h^*-3)} = \frac 14
   \end{equation}
as Theorem
\ref{thm101.2} covers the values $h_{\rm ext} > h^*$.  By default in the sequel
$h_{\rm ext}\in [0,h^*]$ (the analysis of negative magnetic follows by symmetry).

The first step is to use   coarse graining as in Lebowitz-Penrose.

\vskip.5cm

\subsection{The Lebowitz-Penrose procedure}

In this subsection we use the Lebowitz-Penrose
procedure to reduce to a partition
function where the Kac potential is absent.  Let us
first recall the Lebowitz-Penrose result and consider
the partition function
$Z_{\ga, h_{\rm ext},L}^{\rm per}$
with the same short range, vertical interaction as in our case (the ``reference system'' in the Lebowitz-Penrose terminology) but with a Kac potential which has support
on a region of full dimension ($d=2$).  After coarse graining and exploiting
(i) the smoothness
of the Kac potential, (ii) the ferromagnetic nature of the Kac potential, Lebowitz-Penrose have proved that  $Z_{\ga, h_{\rm ext},L}^{\rm per}$ has
the same ``asymptotics'' as
\begin{equation*}
Z^{\rm max}_{\Delta}:=\max_{ m\in \mathcal M_\Delta} e^{(h_{\rm ext}m+ m^2/2) |\Delta|} \sum_{\si \in \{-1,1\}^{\Delta}} e^{-H^{\rm vert}_\ell(\si) } \mathbf 1_{\sum_{x\in \Delta} \si(x)=m|\Delta|}
   \end{equation*}
where $\Delta$ is a square of side $\ell$, $\ell$ the integer part of
$\ga^{-1/2}$, and
\[
\mathcal M_\Delta = \{ -1, -1+ \frac{2}{|\Delta|},\dots,1- \frac{2}{|\Delta|},  1 \}
\]
the set of all possible values of the empirical spin magnetization in   $\Delta$.

By  same ``asymptotics'' we mean that
 \begin{equation}
    \label{F4.1}
 \lim_{\ga\to 0}\lim_{|\La|\to \infty} \frac 1{|\La|} \log Z_{\ga, h_{\rm ext},L}^{\rm per}=
 \lim_{|\Delta|\to \infty} \frac 1{|\Delta|}\log Z^{\rm max}_{\Delta}
   \end{equation}
The same procedure works in our case as well leading to
Theorem \ref{Fthm4.1} below whose proof is given in Appendix \ref{app.B}.

\medskip

\begin{thm}
\label{Fthm4.1}
Let $\Delta$ and $\ell$ be as above,  $\dis{\mathcal M_\ell = \{ -1, -1+ \frac{2}{\ell},\dots,1- \frac{2}{\ell},  1 \}}$,
$m_\Delta(x,i)$, $(x,i)\in \Delta$, a function with values in $\mathcal M_\ell$ which depends only on
$i$,
\begin{equation}
    \label{F4.3}
 \phi_\ell(m_\Delta) = -\frac {1}{|\Delta|}\log \sum_{\si\in \{-1,1\}^\Delta}
  e^{-H^{{\rm vert}}_{\ell}(\si)}\;
  \mathbf 1_{\sum_x \si(x,i) =  m_\Delta(\cdot,i)\;\text{for all $i$}}
   \end{equation}
Then there is $m_+ \in (0,1)$ so that
\begin{equation}
    \label{F4.5}
 \lim_{\ga\to 0}\lim_{L\to \infty} \frac 1{|\La|}
 \log Z_{\ga, h_{\rm ext},L}^{\rm per} =
\lim_{|\Delta|\to \infty} \frac 1{|\Delta|} \log
 Z^{\rm max}_{\Delta}
   \end{equation}
where
\begin{equation}
    \label{F4.4}
\log  Z^{\rm max}_{\Delta}:= \max_{ m_\Delta: |m_\Delta(x,i)| \le m_+} {\sum_{(x,i)\in \Delta}
 \{\frac {m(x,i)^2}2 +
 h_{\rm ext}m(x,i)-\phi_\ell(m_{\Delta})\} }
   \end{equation}

\end{thm}

\medskip
\eqref{F4.1} and \eqref{F4.5} are identical but the meaning of $Z^{\rm max}_{\Delta}$
is different in the two cases.  In \eqref{F4.1} it is a max over a scalar $m$
of the canonical partition function with magnetization $m$.  By classical results on the thermodynamic limit this is related to the free energy of the system and one gets a formula as on the right hand side of \eqref{102.3}.  Thus one has essentially finished once he gets
\eqref{F4.1}, in our case instead \eqref{F4.5} is just the beginning of the work.
In fact the variational problem behind \eqref{F4.4} involves a vector $m_\Delta$ in a space whose dimensions diverge in the thermodynamic limit.  Moreover the relation between
$\phi_\ell(m_\Delta)$ and the $d=1$ free energy $f_\la(m)$ which appears in \eqref{102.3} is not evident due to the multi-canonical constraint of fixing the magnetization on each layer.

 \medskip

The picture looks simpler if we replace the multi-canonical ensemble
by a gran canonical ensemble with auxiliary magnetic fields  on each layer:
let then
$\und h=(h_1,..,h_\ell)$ and
 \begin{equation}
    \label{F4.6}
 Z_{\Delta,\und h} = \sum_{\si\in\{-1,1\}^\Delta} e^{-H^{{\rm vert}}_{\ell}(\si)
 -\sum_{(x,i)\in \Delta}\ h_i  \si(x,i) }
   \end{equation}
The goal is to rewrite  $\phi_\ell(m_\Delta)$ in terms $\log Z_{\Delta,\und h}$ thus proving
an extended version of the equivalence of ensembles theorem. The
first step in this direction is to get a full understanding of
$Z_{\Delta,\und h}$ as provided by the cluster expansion.

\vskip.5cm

\subsection{Cluster expansion}

We first observe that
   \begin{eqnarray}
    \label{F4.10}
 &&\log Z_{\Delta,\und h} = \ell \log Z_{\ell,\und h }\nn\\&&
 Z_{\ell,\und h } = \sum_{\si \in \{-1,1\}^{[1,\ell]}} e^{ \sum_{i=1}^\ell\{\la \si(i)\si(i+1)+h_i\si(i)\}},\quad
 \si(\ell+1)=\si(1)
 \end{eqnarray}
with $Z_{\ell,\und h }$ the partition function of the $d=1$ Ising model with nearest neighbor interactions of strength $\la$ and space dependent magnetic field $\und h$.
We  define
 \begin{eqnarray}
    \label{103.1}
 Z^*_{\ell,\und h }&:=&Z_{\ell,\und h }\;\{\prod_{i=1}^\ell ( e^{h_i}+e^{-h_i})\}^{-1}\\
    \label{103.2}
 u_i &:=& \tanh\{ h_i\}
   \end{eqnarray}
In   Appendix \ref{app.C} we   shall suppose $\la$ small  and use cluster expansion to prove:

\medskip

\begin{thm}
\label{thm103.1}
For any $\la>0$ small enough 
 \begin{equation}
    \label{103.3}
\log Z^*_{\ell,\und h }= \sum_{N(\cdot)}
A_{N(\cdot)} u ^{N(\cdot)}
   \end{equation}
where $N(\cdot): [1,\ell]\to \mathbb N$ and
    \begin{equation}
    \label{103.4}
 u ^{N(\cdot)} =\prod_{i=1}^\ell u_i^{N(i)}
   \end{equation}
The coefficients $A_{N(\cdot)}$ satisfy the following bounds. Call
 \begin{equation}
    \label{103.7}
e^b:= \la^{-5/12},\quad |N(\cdot)|= \sum_x N(x),\quad \|N(\cdot)\| = \max\{|N(\cdot)|,R(N(\cdot))\}
   \end{equation}
where $R(N(\cdot))$ denotes the cardinality of the support of $N(\cdot)$ (i.e.\
the smallest interval
which contains the set  $\{ i: N(i)>0 \}$). Then
for any  $i\in [1,\ell]$ and any positive integer $M$ 
 \begin{equation}
    \label{103.8}
\sum_{N(\cdot):N(i)>0,\|N(\cdot)\|\ge M} |A_{N(\cdot)}| \le e^{-  b  M }
   \end{equation}
Moreover
$A_{N(\cdot)}=0$ if $|N(\cdot)|$ is odd and there are
coefficients $\alpha_{k}$, $k>0$,  and $c$ so that
  \begin{equation}
    \label{103.5}
\sum_{N(\cdot): |N(\cdot)|=2}A_{N(\cdot)}u^{N(\cdot)}=\sum_{i<j }   \alpha_{j-i} u_iu_j
   \end{equation}
 \begin{eqnarray}
    \label{103.6}
&&|\alpha_{1} - \la| \le c \la e^{-2b},\quad
 |\alpha_{j-i}| \le c \la^{|i-j|}e^{|i-j|}
   \end{eqnarray}

\end{thm}

\vskip.5cm

\subsection{Equivalence of ensembles}

The magnetizations associated to $Z_{\ell,\und h }$, as defined in \eqref{F4.10}, are $\und m=(m_1,..,m_\ell)$
\begin{equation}
    \label{1.2}
m_i = 
\frac{\partial}{\partial h_i} \log Z_{\ell,\und h }
   \end{equation}
which are thus expressed via $\und h$ in terms of $(u_1,..,u_\ell)$.
We write more explicitly \eqref{1.2} as
 \begin{equation}
    \label{4.1}
m_{i} = u_i + \Psi_i(u),\quad
\Psi_i(u) = (1-u_i^2)
\sum_{N(\cdot): N(i)>0}N(i)
A_{N(\cdot)} u ^{N^{(i)}(\cdot)}
   \end{equation}
with   $N^{(i)}(k)= N(k)$ for $k\ne i$ and $N^{(i)}(i)= N(i)-1$.
In Appendix \ref{app.D} we will prove that there is a one to one correspondence
between $\und u$ and $\und m$ so that we may write
$\und u$ as a function of $\und m$.

\medskip

\begin{thm}
\label{Fthm5.6}

For any $\la>0$ small enough the following holds.  For any $\und m$ such that $|m_i| \le m_+$ ($m_+$ as in \eqref{F4.4}) there is a unique $\und h$
such that \eqref{1.2} holds for any $i=1,..,\ell$ and there exists $h_+>0$ so that all the components of $\und h$ are bounded by $h_+$.


\end{thm}

\medskip

\medskip

\begin{thm}
\label{thm101.4}
For any $\la>0$ small enough the following holds. For any $m_\Delta=\und m=(m_1,..,m_\ell)$, $|m_i| \le m_+$, $i=1,..\ell$, call $\und h=(h_1,..,h_\ell)$
the magnetic fields associated to $\und m$ via Theorem \ref{Fthm5.6}, then
for any $a \in (\frac 12, 1)$
there is $c$ so that
 \begin{equation}
    \label{101.24}
   \Big|  \frac{1}{\ell^2}\log \{e^{-\ell \sum_i h_i m_i}Z_{\Delta,\und h}\}
+ \phi_\ell ( m_\Delta)\Big| \le c \ell^{a-1}
   \end{equation}
where $\phi_\ell (m_\Delta)$ is defined in \eqref{F4.3} and $Z_{\Delta,\und h}$
in \eqref{F4.6}.

\end{thm}

As a consequence:

\medskip

\begin{thm}
\label{Fthm4.2}
Let $Z^{\rm max}_{\Delta}$ be as in \eqref{F4.4}, then
   \begin{equation}
    \label{F4.7}
\lim_{|\Delta|\to \infty} \frac 1{|\Delta|}\log
 Z^{\rm max}_{\Delta} = \lim_{\ell\to\infty}
 \frac 1 \ell\max _{\und h: |h_i| \le h_+}\sum_{i=1}^\ell [ \frac{m_i^2}2+
 (h_{\rm ext}-h_i)m_i + \log Z_{\ell,\und h }]
    \end{equation}
where $m_i$ is the function of $\und h$ defined in \eqref{1.2}--\eqref{4.1}.

\end{thm}

\medskip

By \eqref{F4.5} and \eqref{F4.7} we get
 \begin{equation}
    \label{F4.9}
 \lim_{\ga\to 0}\lim_{L\to \infty} \frac 1{|\La|}
 \log Z_{\ga, h_{\rm ext},L}^{\rm per} = \lim_{\ell\to\infty}
 \frac 1 \ell\max _{\und h}\sum_{i=1}^\ell [ \frac{m_i^2}2+
 (h_{\rm ext}-h_i)m_i + \log Z_{\ell,\und h }]
   \end{equation}

\vskip.5cm

\subsection{The quadratic structure of the effective hamiltonian}
\label{Fsubsec.4.4}

The goal therefore  is to study the  ground state energy of the effective hamiltonian
 \begin{equation}
    \label{F4.11}
H^{\rm eff}_{\ell,\und h } = -\sum_{i=1}^\ell \{\frac {m_i^2}2 - h_im_i + h_{\rm ext} m_i\}  -\log(e^{h_i}+e^{-h_i})-\log Z^*_{\ell,\und h } +   A_{0}
   \end{equation}
regarded as a function of $\und u
=(u_1,..,u_\ell)$. For convenience
in \eqref{F4.11} we have subtracted to $\log Z^*_{\ell,\und h }$ (which is defined in \eqref{103.3}) the first term of the expansion \eqref{103.3} (with $N(\cdot)\equiv 0$), which
is a constant.

By Theorem \ref{Fthm5.6} we can restrict to  the set of $\und u: |u_i| \le u_+=\tanh(h_+), i=1,..,\ell$ and in the sequel we will tacitly restrict to such a set. We will prove that the inf over $\und u$ of
$H^{\rm eff}_{\ell,\und h }$ is achieved by vectors $\und u$ with all components equal to each other which is maybe the most relevant/original result of this paper.

 We start by making explicit the leading terms in \eqref{F4.11}
for $\la$ small.  To this end and recalling that $\log(e^{h_i}+e^{-h_i}) = h_iu_i +S(u_i)$, the entropy $S(u)$ being defined in \eqref{101.4.00}--\eqref{101.4.00.01},
we write
     \begin{eqnarray}
    \label{2.24}
 &&\log(e^{h_i}+e^{-h_i}) = h_iu_i + \frac {u_i^2}2 +T(u_i)\nn\\&&
  T(u) =
 -\log 2 +\sum_{k=1}^\infty \frac 1{2k+1}\frac 1{2k+2} u^{2k+2}\\
     \label{Z.1.1}
&&\log Z^*_{\ell,\und h}-A_{0}= -\frac{\la}2\sum_{i=1}^\ell(u_{i+1}-u_i)^2 + \Theta\\
    \label{Z.1.2}
 &&  \xi_i:= (h_i-u_i) (1-u_i^2)\\
    \label{Z.1}
&&\Psi_i  = \la (1-u_i^2) (u_{i+1}+u_{i-1})
 + \Phi_i
   \end{eqnarray}
where $\Psi_i$ is defined in \eqref{4.1} and $\Theta$ and $\Phi_i$   are defined by
the above equations.
By some simple algebra, see  Appendix \ref{FApp.E} for details, we can rewrite  $H^{\rm eff}_{\ell,\und h }$ as:

\medskip

   \begin{lem}
   \label{Flemma4.4.1}
With the above notation:
\begin{eqnarray}
    \label{Z.3}
 H^{\rm eff}_{\ell,\und h }  &=&\sum_{i=1}^\ell \left\{ T(u_i) - h_{\rm ext}u_i -2\la  h_{\rm ext}[u_i- u^3_i]
+2\la \xi_i u_i\right\}\nn\\&+&
  \sum_{i=1}^\ell \left\{\frac \la 2(u_i-u_{i+1})^2+(h_i-u_i)\Phi_i - h_{\rm ext} \Phi_i -\frac {\Psi_i^2}2  \right\} - \Theta\nn\\&-&\la\sum_{i=1}^\ell \left\{ h_{\rm ext}(u_i +u_{i+1})(u_{i+1}-u_i )^2
 + (\xi_i-\xi_{i+1})(u_i-u_{i+1})\right\}
   \end{eqnarray}
   \end{lem}


\medskip
The terms with $\Theta$, $\Psi_i^2$ and $\Phi_i$ are ``under control'' in the following sense:

\medskip

\begin{thm}
\label{Fthm.4.8}
Call
\begin{eqnarray}
    \label{Z.3.01}
 H^{(1)}_{\ell,\und h }  &=&
  \sum_{i=1}^\ell \left\{ (h_i-u_i)\Phi_i - h_{\rm ext} \Phi_i -\frac {\Psi_i^2}2  \right\} - \Theta
   \end{eqnarray}
Then for any $h_{\rm ext} \in [0,h^*]$ there is a continuous function $f^{(1)}(u)$ on $[-1,1]$ and functions
$b^{(1)}_{i,j}(\und u)$, $i<j$, so that
     \begin{eqnarray}
    \label{F4.23}
&&
H^{(1)}_{\ell,\und h } = \sum_{i=1}^\ell f^{(1)}(u_i) + \sum_{1\le i<j\le \ell}
b^{(1)}_{i,j}(\und u) (u_i-u_j)^2
   \end{eqnarray}
with
     \begin{eqnarray}
&&
\sum_{1\le i<j\le \ell}
b^{(1)}_{i,j}(\und u) (u_i-u_j)^2 \ge -c \la^{1+\frac{2}{3}}\sum_{i=1}^\ell  (u_i-u_{i+1})^2
  \label{F4.24}
   \end{eqnarray}
and $c$ a positive constant.
\end{thm}

\medskip

The proof of Theorem \ref{Fthm.4.8} starts from \eqref{103.3} and it is based on a representation of the monomials $u^{N(\cdot)}$ as sum of one body and gradients squared terms which is
established in Section \ref{Fsec.5}.  After that we exploit the properties of the coefficients
$A_{N(\cdot)}$ stated in \eqref{103.8},  \eqref{103.5} and  \eqref{103.6}. The computations are straightforward but lengthy, the details are reported in  Appendix \ref{FApp.F}.

Define
\begin{eqnarray}
    \label{Z.3.02}
\theta_i(\und u) := \frac{\xi_i-\xi_{i+1}}{u_i-u_{i+1}}
   \end{eqnarray}
when $u_i\ne u_{i+1}$ and equal to $d\xi/du(v)$ when $u_i=u_{i+1}=v$.
Then  $\theta_i(\und u)$ depends only on $u_i$ and $u_{i+1}$,
and as a function  of  $u_i$ and $u_{i+1}$ is continuous, symmetric and bounded in $|u_i|\le R$,
$|u_{i+1}| \le R$, $R<1$. It is essential in our proof that $\theta_i(\und u)< \frac 12$.  We checked
numerically that this is ``fortunately'' true and indeed we found a rigorous proof, reported
in Appendix \ref{FApp.G}, of an upper bound smaller than $1/2$:

\medskip

\begin{prop}
\label{Fprop.4.8.1}

     \begin{eqnarray}
    \label{Z.3.03}
  \theta_i(\und u) &\le& \frac 38
   \end{eqnarray}

\end{prop}

\medskip
It follows from \eqref{Z.3} and \eqref{F4.23} that
\begin{eqnarray}
    \label{Z.3.04}
 H^{\rm eff}_{\ell,\und h }  &=&\sum_{i=1}^\ell f(u_i)
+  \la\sum_{i=1}^\ell \left\{\frac 1 2-h_{\rm ext}(u_i +u_{i+1})-\theta_i(\und u)\right\}(u_i-u_{i+1})^2
\nn\\&+& \sum_{i<j} b^{(1)}_{i,j}(\und u) (u_i-u_j)^2
   \end{eqnarray}
where
\begin{eqnarray}
    \label{Z.3.05}
 &&f(u_i) =\{ T(u_i) - h_{\rm ext}u_i -2\la  h_{\rm ext}[u_i- u^3_i]
+2\la \xi_i u_i\}+ f^{(1)}(u_i)
   \end{eqnarray}
Let $u^*$ be the minimizer of $f(\cdot)$, then
by \eqref{Z.3.04}
\[
\inf_{\und u}H^{\rm eff}_{\ell,\und h } \le \ell \inf f(u) =: \ell f(u^*)
\]
as the right hand side is the value obtained by choosing all $u_i = u^*$.

Let
 \begin{equation}
   \label{Z.3.06}
  h_0:=\frac 14 \left[\frac 12 -\frac 38 \right]
   \end{equation}
then if $h_{\rm ext} \in [0,h_0]$ using \eqref{F4.24} we get
\begin{eqnarray}
    \label{Z.3.07}
 H^{\rm eff}_{\ell,\und h }  &\ge&\sum_{i=1}^\ell f(u_i)
+ \frac \la 2\sum_{i=1}^\ell \left\{\frac 1 2-\frac 38 \right\}(u_i-u_{i+1})^2
-\sum_{i } c \la^{1+\frac{2}{3}} (u_i-u_{i+1})^2 \nn\\&\ge&
\sum_{i=1}^\ell f(u_i)
+ \frac \la 2\sum_{i=1}^\ell \left\{\frac 1 2-\frac 38- 2c\la^\frac{2}{3}\right\}(u_i-u_{i+1})^2
   \end{eqnarray}
Hence
\[
\inf_{\und u}H^{\rm eff}_{\ell,\und h } \ge   \ell f(u^*)
\]
for $\la$ so small that $2c\la^{\frac{2}{3}} \le \frac 12 - \frac 38$, because by \eqref{Z.3.07} we then get a lower bound by neglecting the sum of the terms with $(u_i-u_{i+1})^2$.  We have thus proved:

\medskip

\begin{thm}
\label{Fthm.4.8.1}
Let   $h_{\rm ext} \in [0,h_0]$ and $\la$ be
so small that $2c\la^{\frac{2}{3}} \le \frac 12 - \frac 38$, then
the inf of $H^{\rm eff}_{\ell,\und h }$ is equal to the min
of $H^{\rm eff}_{\ell,\und h }$ over homogeneous $\und h$,
namely $\und h$ with all its components
equal to each other.

\end{thm}

\medskip

Thus the inf  of $H^{\rm eff}_{\ell,\und h }$ is achieved when
 all the components of $\und h$ are
equal to each other, in such a case we can compute
explicitly the minimizer, see the next subsection.
The result comes from the quadratic structure
of the hamiltonian, \eqref{F4.23}--\eqref{F4.24}, somehow reminiscent of the Ginzburg-Landau functional whose integrand has the form $W(u) + |\nabla u|^2$ and the gradient term forces the minimizer to be a constant.

The argument used to prove Theorem \ref{Fthm.4.8.1} does not extend to the complementary case when
$h_{\rm ext} \in [h_0,h^*]$ because   $\{\frac 1 2-h_{\rm ext}(u_i +u_{i+1})\}$ in \eqref{Z.3.04}
may become negative and we would then
loose the positivity of the coefficients of the gradients squared.
Nonetheless we can use the  ``strong convexity'' of the one body term $T(u_i)$ in \eqref{Z.3}
when $u_i$ is away from 0 to prove:

\medskip

\begin{thm}
\label{Fthm.4.8.2}
Let   $h_{\rm ext} \in [h_0,h^*]$ and let $\la$ be small enough, then
the inf of $H^{\rm eff}_{\ell,\und h }$  is equal to the min
of $H^{\rm eff}_{\ell,\und h }$ over homogeneous $\und h$.

\end{thm}

 Theorem \ref{Fthm.4.8.2} is proved in Appendix \ref{FApp.H}

\vskip.5cm

\subsection{Proof of  Theorem \ref{thm101.1}}
\label{Fsubsec.4.1}
Using Theorem  \ref{Fthm.4.8.1} and Theorem  \ref{Fthm.4.8.2}
we will next
prove  Theorem \ref{thm101.1}.
We thus know that
 \begin{equation}
    \label{101.33}
 \lim_{\ga\to 0}\lim_{L\to \infty}  \frac 1{|\La|}
 \log Z_{\ga, h_{\rm ext},L}^{\rm per} =
 \lim_{\ell\to \infty}
  \max _{m \in \mathcal M_\ell:|m| \le m_+}  [\frac{m^2}2+
 h_{\rm ext} m - \psi_{\la,\ell}(m)]
   \end{equation}
where, letting $\Delta = I\times I'$,
\begin{equation}
    \label{102.88}
 \psi_{\la,\ell}(m) = -\frac {1}{|\Delta|}\log \sum_{\si\in \{-1,1\}^\Delta}
  \mathbf 1_{\sum_{x\in I} \si(x,i) =  \ell m, \;\text{for all}\,i\in I'} e^{-H^{{\rm vert}}_\ell(\si)}
   \end{equation}
Denote by
\begin{equation}
    \label{102.88.1}
 f_{\la,\ell}(m) = -\frac {1}{|\Delta|}\log \sum_{\si\in \{-1,1\}^\Delta}
  \mathbf 1_{\sum_{(x,i)\in \Delta} \si(x,i) =  |\Delta| m} e^{-H^{{\rm vert}}_\ell(\si)}
   \end{equation}
the finite volume free energy of the system with only vertical interactions.  Thus
$\lim_{\ell \to \infty}f_{\la,\ell}(m) =f_{\la}(m)$.  We obviously have
$-\psi_{\la,\ell}(m) \le -f_{\la,\ell}(m)$ and by classical results on the Ising model
\begin{equation}
    \label{102.88.2}
 -f_{\la,\ell}(m) \le -f_{\la}(m) +\frac {c}{\ell}
   \end{equation}
so that
\begin{equation}
    \label{102.88.3}
 \limsup_{\ga\to 0}\lim_{L\to \infty} \frac 1{|\La|}
 \log Z_{\ga, h_{\rm ext},L}^{\rm per} \le  \lim_{\ell\to \infty}
  \max _{m\in \mathcal M_\ell: |m| \le m_+}  [ \frac{m^2}2+
 h_{\rm ext} m - f_{\la}(m)]
   \end{equation}
which proves that the left hand side of \eqref{102.3} is bounded by its right hand side.

We are thus left with the proof of the reverse inequality.
Let
 \begin{eqnarray}
    \label{102.88.4}
&&\tilde m = \text{arg min}\Big\{- h_{\rm ext} m -\frac{m^2}2 + f_\la(m)\Big\}\\&&
m^{(\ell)} = \max\Big\{ m\in \mathcal M_\ell : m \le \tilde m \Big\}\nn\\&&
h^{(\ell)} : \; \frac{d}{dh} p_{\la,\ell}(h)\Big |_{h=h^{(\ell)}} = m^{(\ell)}
   \end{eqnarray}
where $p_{\la,\ell}(h)= \ell^{-1}\log Z_{\la,h,\ell}$ and $Z_{\la,h,\ell}$
is the partition function of  the Ising model in $[1,\ell]$ with
n.n. interaction of strength $\la$ and magnetic field $h$;
$p_\la(h)$ is the corresponding  pressure in the thermodynamic limit $\ell \to \infty$.  It is well known that
\begin{equation}
    \label{102.88.7}
\sup_h | p_{\la,\ell}(h)- p_{\la}(h)| \le \frac c \ell
   \end{equation}
  Then
\begin{equation}
    \label{102.88.5}
 \Big\{h_{\rm ext} m^{(\ell)} +\frac{(m^{(\ell)})^2}2 -h^{(\ell)}m^{(\ell)}+
 \frac 1{|\Delta|}\log\Big( \sum_\si \mathbf 1_{\sum_x \si(x,i)\equiv \ell m^{(\ell)}}
 e^{-H^{\rm vert}_\ell(\si) + \sum h^{(\ell)}\si(x,i)}\Big\}
   \end{equation}
is a lower bound in the asymptotic sense for $ \frac 1{|\La|}
 \log Z_{\ga, h_{\rm ext},L}^{\rm per}$.
By the equivalence of ensembles, see Theorem \ref{thm101.4} in Subsection \ref{Fsubsec.3.3}, the lower bound becomes
\begin{equation}
    \label{102.88.6}
 \Big\{h_{\rm ext} m^{(\ell)}+\frac{(m^{(\ell)})^2}2 -h^{(\ell)}m^{(\ell)}+
 p_{\la,\ell}(h^{(\ell)})\Big\}
   \end{equation}
where by \eqref{102.88.7} we can also replace $p_{\la,\ell}(h^{(\ell)})$ by $p_{\la}(h^{(\ell)})$.
By compactness $h^{(\ell)}$ converges by subsequences and if $\ell_k$ is a convergent subsequence there is $h$ so that $h^{(\ell_k)} \to h$ and consequently
   \begin{equation}
    \label{102.88.8}
\lim_k p'_{\la,\ell}(h^{(\ell_k)})= p'_\la(h)
   \end{equation}
In fact in general any limit point of $ p'_{\la,\ell}(h^{(\ell_k)})$ lies in the interval
$[\frac{d}{dh^-}p_\la(h),\frac{d}{dh^+}p_\la(h)]$ of its left and right derivatives, but since we are considering a $d=1$ system such derivatives are equal to each other.  Moreover by the choice of $h^{(\ell)}$
   \begin{equation}
    \label{102.88.8}
 p'_{\la,\ell}(h^{(\ell)})= m^{(\ell)} \to \tilde m
   \end{equation}
 Hence if $h^{(\ell_k)} \to h$ then
      \begin{equation}
    \label{102.88.9}
 p'_\la(h) = \tilde m
   \end{equation}
 and since there is a unique $\tilde h$ such that   $ p'_\la(\tilde h) = \tilde m$ it follows that for any convergent subsequence $h^{(\ell_n)} \to \tilde h$
 and therefore $h^{(\ell)} \to \tilde h$.  Thus the expression \eqref{102.88.6} converges to
         \begin{equation}
    \label{102.88.10}
 \Big\{h_{\rm ext} \tilde m +\frac{(\tilde m)^2}2 -\tilde h \tilde m+
 p_{\la}(\tilde h)\Big\}
   \end{equation}
   which concludes the proof because
   $-\tilde h \tilde m+
 p_{\la}(\tilde h) = f_\la(\tilde m)$.

\vskip1cm

\setcounter{equation}{0}

\section{Monomials are sum of gradients}
\label{Fsec.5}

In this section we prove a combinatorial lemma, Theorem \ref{Fthm.5.9} below, which is the key ingredient in
the proof of the gradient structure of the hamiltonian.  The whole section is self contained
and can be read independently of the rest of the paper.


\medskip

 \begin{thm}
 \label{Fthm.5.9}
 Let  $\und u=(u_1,..,u_k)\in \mathbb R^k$, $\und
n= (n_1,..,n_k)\in\mathbb N_+^k$ and
   \[
 M_{\und n}( \und u) 
 =u_1^{n_1} \cdots u_k^{n_k},\quad \sum_{i=1}^k n_i =: N
   \]
a monomial of degree $N$ in the $k$ real variables $u_1,..,u_k$.  Then for any $N\ge 2$
 \begin{equation}
    \label{F5.1}
M_{\und n}(\und u) = \sum_{i=1}^k p_i u_i^N + \sum_{1\le i<j\le k} d_{i,j}(\und u) (u_i-u_j)^2
   \end{equation}
where $(p_1,..,p_k)$ is a probability vector, its component $p_i$ depending on $\und n$;
$d_{i,j}(\und u)$ are polynomials of degree $N-2$   with negative coefficients
which  depend on $\und n$ and there is a constant $c$ so that for any positive $U\le 1$
 \begin{equation}
    \label{F5.2}
\sup_{|u_i| \le U, i=1,..,k}\;|d_{i,j}(\und u) | \le c U^{N-2}N^3
   \end{equation}

 \end{thm}

\medskip

\begin{proof}
Call  $N_j= n_1+\dots+n_j$, $j=1,..,k$, so that $N_k=N$.  We will
prove the theorem with
 \begin{equation}
    \label{F5.3}
    d_{i,j}(\und u) = \sum_{m=1}^{N_j-1} c_{i,j;m} \Big(u_{i}^{m-1}
u_{j}^{N_j-m-1} u_{{j+1}}^{n_{j+1}}\cdots u_{k}^{n_k}\Big)
   \end{equation}
with coefficients  $c_{i,j;m}$, $i<j$, $1\le m\le N_j-1$, which depend on $n_1,..,n_j$ and
satisfy the bound
\begin{equation}
    \label{F5.4}
|c_{i,j;m}| \le C N_j^2
   \end{equation}
\eqref{F5.2} follows from \eqref{F5.3}--\eqref{F5.4} which also show that the representation \eqref{F5.1} of $M_{\und n}(\und u)$ is not unique as we can commute the factors $u^{n_i}_i$ in the monomial $M_{\und n}(\und u)$ without changing its value.

The proof of \eqref{F5.1} generalizes the equality
\[
 uv = \frac {u^2}{2} + \frac {v^2}{2} - \frac {1}{2} \big(u-v\big)^2
\]
In fact we use the above identity to rewrite the factor $u_1u_2$ in $M_{\und n}\equiv
M_{\und n}(\und u)$ getting
\begin{equation}
    \label{F5.4.1}
 M_{\und n}   = \frac {1}{2} M_{\und n + \und e_1 - \und e_2}  +
 \frac {1}{2} M_{\und n - \und e_1 + \und e_2} - \frac {1}{2} M_{ \und n - \und e_1 - \und e_2} \big(u_1-u_2\big)^2
   \end{equation}
where
\[
 \und e_1 =(1,0,..,0),\quad   \und e_2=(0,1,0,..,0)
 \]
We thus have
\[
  \frac {1}{2} M_{\und n + \und e_1 - \und e_2}  +
 \frac {1}{2} M_{\und n - \und e_1 + \und e_2} -M_{\und n} = \frac {1}{2} M_{ \und n - \und e_1 - \und e_2} \big(u_1-u_2\big)^2
 \]
which reminds of the discrete version of the equation $\Delta u=f$ that we will
solve by iteration.  There is   a nice probabilistic interpretation under which
the terms $\frac {1}{2} M_{\und n + \und e_1 - \und e_2}$ and $\frac {1}{2} M_{\und n - \und e_1 + \und e_2}$ will be interpreted by saying that $n_1\to n_1\pm 1$ with probability $1/2$,
see the process $\und n(t)$ defined below.
With this in mind we introduce a Markov chain
$\xi(t), t\in \mathbb N$, $\xi(t)\in \Om$, where
 \begin{eqnarray}
    \label{3.4}
 &&   \Om = \bigcup_{i=1}^{k-1} \Om_i,\quad
\Om_i = \{ (i,x): 1\le x \le N_{i+1}-1\},\quad i< k-1  \nn\\&&
\hskip2cm \Om_{k-1}=
 \{(k-1,x): 0 \le x \le N_k\}
   \end{eqnarray}
The transition probabilities $P(\cdot,\cdot)$ are all set equal to 0 except
 \begin{eqnarray}
    \label{3.5.0}
&& P((i,x),(i,y) )= \frac 12, \; |x-y|=1;\nn
\\&&
P((i,1),(i+1,N_{i+1}) )= P((i,N_{i+1}-1),(i+1,N_{i+1}) )= \frac 12
\nn\\&&
P((k-1,0),(k-1,0))=P((k-1,N_k),(k-1,N_k)) =1
   \end{eqnarray}
 The first line in \eqref{3.5.0} describes the motion on the components $\Om_i$ of $\Om$;
 the second one the jump from $\Om_i$ to $\Om_{i+1}$ (the reverse jump having 0 probability)
 while the last line says that the endpoints of $\Om_{k-1}$ are ``traps'', namely once the chain reaches those points it is stuck there forever.

We start the chain at time 0 from
 \begin{equation}
    \label{3.6}
\xi(0) = (1,n_1)
   \end{equation}
We call $\tau_i$, $i=0,..,k-2$, the first time $t$ when $\xi(t)\in \Om_{i+1}$ ($\tau_0=0$) and define for $i\ge 1$, $\si_i= \pm$ if $\xi(\tau_i-1)=N_{i+1}$, respectively $\xi(\tau_i-1)=1$.
For $\tau_i \le t<\tau_{i+1}$ $\xi(t)$ is
a simple symmetric random walk,   thus, by classical estimates,
there are constants $b>0$, $c>0$ so that
 \begin{equation}
    \label{3.7}
P[\tau_{i+1}-\tau_i > s] \le c e^{-b N_{i+2}^{-2}s}
   \end{equation}
%
%
%

To establish a connection with $M_{\und n}(\und u)$ and the iterates of \eqref{F5.4.1}
we introduce new processes $\und n(t)$, $f(t)$ and $a(t)$ which are all
``adapted to the canonical filtration $\mathcal F_t$'', calling
$\theta(t)$   adapted to the canonical filtration $\mathcal F_t$
if $\theta(t)$ is determined by $\{\xi(s), s\le t\}$.

Let $\und n(t)=(n_1(t),..,n_k(t))$ be defined as follows.  When $t < \tau_1$
we set
 \begin{equation}
    \label{3.8}
\und n(t) = \big(\xi(t),N_2-\xi(t),n_3,..,n_k \big)
   \end{equation}
For $t=\tau_1$:
\begin{equation}
    \label{3.8}
\und n(\tau_1) = \begin{cases}\big(N_2,0,n_3,..,n_k \big) &\text{if $\si_1=+$}
\\\big(0,N_2, n_3,..,n_k \big) &\text{if $\si_1=-$}\end{cases}
   \end{equation}
In the interval $\tau_1 \le t < \tau_2$
\begin{equation}
    \label{3.9}
\und n(t) = \begin{cases}\big(\xi(t),0,N_3-\xi(t),n_4,..,n_k \big) &\text{if $\si_1=+$}
\\\big(0,\xi(t), N_3-\xi(t),n_4..,n_k \big) &\text{if $\si_1=-$}\end{cases}
   \end{equation}
By iteration the definition is extended to all $t\in \mathbb N$.  The process $\und n(t)$ is indeed quite simple: fix $2\le i\le k$, then  $n_i(t)=n_i(0)$ for $t\le \tau_{i-2}$ after that it performs a simple symmetric random walk with absorption at 0.  In the time
interval $\tau_{i-2} \le t \le \tau_{i-1}$ all $n_j(t)=0$ with $j<i$ except one, whose label
is denoted $\ell_i$, which  jumps   with opposite sign as $n_i(t)$. For $j>i$,
$n_j(t)=n_j(0)$ in $\tau_{i-1} \le t \le \tau_i$.

The process $f(t)$ is defined as
\begin{equation}
    \label{3.10}
f(t) = \prod_{i=1}^k u_i^{n_i(t)}
   \end{equation}
while $a(t)$  is defined by setting
\begin{equation}
    \label{3.11}
a(t) =  \frac 12\{u_{\ell_i}^{n_{\ell_i}(t)-1}
u_{i}^{N_{i}-n_{\ell_i}(t)-1}
\prod_{j>i} u_j^{n_j}\}
\Big(u_{\ell_i} - u_{i} \Big)^2,\quad
 \tau_{i-2} \le t <\tau_{i-1}
   \end{equation}

 We are going to prove that $f(0) =M_{\und n}(\und u)=u_{1}^{n_1}\cdots u_{k}^{n_k}$ is equal to
\begin{equation}
    \label{3.12}
f(0) = E\Big[ f(t)  \Big] - \sum_{s=-1}^{t-1} E\Big[ a(s)  \Big],\quad a({-1})=0,\;t\ge 0
   \end{equation}
\eqref{3.12} will be proved by showing that
\[
m(t):=f(t)-\sum_{s=-1}^{t-1}   a(s) \; \;\text{is a $\mathcal F_t$-martingale}
\]
namely that
\[
E[m(t+1)\,|\,\mathcal F_t]=0,\quad E[f(t+1)-f(t)\,|\,\mathcal F_t]= a(t)
\]
Since we are conditioning on $\mathcal F_t$ we know the process till time $t$,
suppose that 
$\tau_{i-2} \le t < \tau_{i-1}$, call $\und n(t)= (n'_1,...,n'_{i+1}, n_{i+2},..,n_k)$,
so that all   $n'_j=0$ with $j<i$ except $\ell_i$,
while $n'_j=n_j(0)=n_j$ for $j>i$.
%
Then, by \eqref{3.10}, 
\[
f(t) = u_{\ell_i}^{n'_{\ell_i}}u_{i}^{n'_{i}} \prod_{j=i+1}^k
u_j^{n_j}
\]
and
\[
E[f(t+1)\,|\,\mathcal F_t]= \frac 12 \{u_{\ell_i}^{n'_{\ell_i}+1}u_{i}^{n'_{i}-1}
+u_{\ell_i}^{n'_{\ell_i}-1}u_{i}^{n'_{i}+1}\} \prod_{j=i+1}^k u_j^{n_j}
\]
so that $E[f(t+1)\,|\,\mathcal F_t]- f(t)\,$ is equal to
\[
 \frac 12 \{u_{\ell_i}^{2}
+ u_{i}^{2}-2 u_{\ell_i} u_{i} \} \{u_{\ell_i}^{n'_{\ell_i}-1}u_{i}^{n'_{i}-1} \prod_{j=i+1}^k u_j^{n_j}\}
=  \frac 12 \big(u_{\ell_i}-u_{i}\big)^{2}  u_\ell^{n'_{\ell_i}-1}u_{i}^{n'_{i}-1} \prod_{j=i+1}^k u_j^{n_j}
\]
which is equal to $a(t)$.  Thus $E[f(t+1)-a(t)\,|\,\mathcal F_t] =0$ and therefore $m(t)$ is a martingale.

Since $P[\tau_{k-1}< \infty]=1$ we can take the limit as $t\to \infty$ in \eqref{3.12} which yields \eqref{F5.1} with
\begin{equation}
    \label{3.13}
 p_1 = \frac 12 P[\si_j = +, j\ge 1],\quad
  p_i = \frac 12 P[\si_{i-1} = -, \si_j = +, j\ge i], \; i>1
   \end{equation}
    \begin{eqnarray}
    \label{3.14}
c_{i,j;m} &=& - \frac 12 P\Big[\si_{i-1} = -, \si_n = +, i\le n \le j-1\Big] \nn\\ &\times&
\sum_{t\ge 0}
P_{N_{j-1}}\Big[x^0(t)=m, x^0(s) \in [1, N_j-1], s \le t\Big]
   \end{eqnarray}
   where $x^0(t)$ is a simple symmetric random walk which starts from
   $N_{j-1}$.

\end{proof}

\vskip1cm

%
%
%
%
%
%
%
%
%

\vskip2cm

\appendix

\setcounter{equation}{0}

\section{Proof of Theorem \ref{thm101.2}}
\label{app.A}
%

We preliminary observe that
for any $h_{\rm ext}>0$ there is $m$ so that
$h_{\rm ext} +m = f'_\la(m)$: in fact  $h_{\rm ext} +m - f'_\la(m)$
is positive at $m=0$ and negative as
$m\to 1$ with $f'_\la(m)$ continuous.
If there are several $m$ for which
the equality holds we arbitrarily fix one of them that we denote
by $m_{h_{\rm ext}}$, we shall see a posteriori that there is
uniqueness.  To compute the left hand side
of \eqref{102.3.00} we introduce an interpolating hamiltonian.  For $t\in [0,1]$ we set:
 \begin{eqnarray}
    \label{102.3.01}
H_{t,\ga,L}  (\si)&=& t H^{\rm per}_{\ga, h_{\rm ext},L}(\si) + (1-t) H^0_L (\si)\\
 H^0_L (\si) &=& H^{\rm vert}_L(\si) +H_{h_{\rm ext},L}- \sum_{(x,i)\in \La} m_{h_{\rm ext}} \si(x,i)\nn
   \end{eqnarray}
Denote by $Z^0_L$ the partition function with  hamiltonian  $H^0_L$, by
$P_{t,\ga,L}$ the Gibbs measure with hamiltonian $H_{t,\ga,L}$  and by
$E_{t,\ga,L}$ its expectation, then
 \begin{equation}
    \label{102.3.02}
\log Z_{\ga, h_{\rm ext},L}^{\rm per} - \log Z^0_L = \int_0^1 E_{t,\ga,L}[
H^0_L-H^{\rm per}_{\ga, h_{\rm ext},L}] dt
   \end{equation}
The thermodynamic limit of $\log Z^0_L/|\La|$ is the pressure of the $d=1$ Ising model with only vertical
interactions and magnetic field $h_{\rm ext} +m_{h_{\rm ext}}$, thus, by the choice of
$m_{h_{\rm ext}}$:
\begin{equation}
    \label{102.3.03}
\lim_{L\to \infty} \frac{\log Z^0_L}{|\La|} =  (h_{\rm ext} +m_{h_{\rm ext}})m_{h_{\rm ext}}
- f_\la(m_{h_{\rm ext}})
   \end{equation}
To compute  the left hand side
of \eqref{102.3.00} we need to control the expectation on the right hand side of
\eqref{102.3.02} that we will do by exploiting the assumptions on $h_{\rm ext}$ which imply
the validity of the Dobrushin uniqueness criterion as
we are going to show.  The criterion involves the Vaserstein distance of the conditional probabilities $P_{t,\ga,L}[ \si(x,i)\;|\; \{\si(y,j)\}]$ of a spin $\si(x,i)$ under different
values of the conditioning spins $\{\si(y,j), (y,j)\ne (x,i)\}$.  In the case of Ising spins such Vaserstein distance is simply equal to the absolute value of the difference of the conditional expectations and the criterion requires that for any pair of spin configurations
outside $(x,i)$
 \begin{eqnarray}
    \label{101.9}
 &&\hskip-.5cm |
E_{t,\ga,L}[ \si(x,i)\;|\: \{\si(y,j)\}]-E_{t,\ga,L}[ \si(x,i)\;|\: \{\si'(y,j)\}]|\nn\\&& \hskip1cm
 \le \sum_{y,j}r_{\ga,L}(x,i;y,j)|\si'(y,j)-\si(y,j)|, \qquad
\sum_{y,j}r(x,i;y,j) \le r <1
   \end{eqnarray}
Since
  \begin{eqnarray*}
&&E_{t,\ga,L}[ \si(x,i)\;|\: \{\si(y,j)\}] =\tanh \Big\{ t\sum_{y\ne x} J_{\ga,L}(x,y) \si(y,i)+(1-t)m_{h_{\rm ext}}\\&&\hskip2cm
+\la[\si(x,i+1)+\si(x,i-1)] + h_{\rm ext}\Big\}
   \end{eqnarray*}
($J_{\ga,L}(x,y)$  the kernel $J_\ga(x,y)$ with periodic boundary conditions in $\La$) one can easily check that \eqref{101.9} is satisfied with $r$  as in \eqref{101.8} and
$r(x,i;y,j)=r_{\ga,L}(x,i;y,j)$ with
 \begin{equation}
    \label{101.9.01}
  r_{\ga,L}(x,i;y,j) = \cosh^{-2}( h_{\rm ext}-1-2\la)\Big( J_{\ga,L}(x,y)\mathbf 1_{j=i}
+\la \mathbf 1_{x=y; j=i\pm 1\;{\rm mod} \, L}\Big)
   \end{equation}

By the Dobrushin uniqueness theorem
there is a unique DLR measure
$P_{t,\ga}$ which is the weak limit of $P_{t,\ga,L}$ as $L\to \infty$.  We denote
by $m_{t,\ga,L}$ and $m_{t,\ga}$ the  average of a spin under  $P_{t,\ga,L}$ and $P_{t,\ga}$.
We call $\nu^0_L$ and $\nu^0$ the measures $P_{t,\ga,L}$ and $P_{t,\ga}$ when $t=0$,
thus $\nu^0_L$ is
the Gibbs measure for the Ising system in $\La$ with hamiltonian $H^{\rm vert}$ and magnetic field
$h_{\rm ext}+m_{h_{\rm ext}}$,  $\nu^0$ denoting its thermodynamic limit.  We then have
 \begin{equation}
    \label{101.9.01.1}
 \lim_{L\to\infty} m_{t,\ga,L} = m_{t,\ga},\qquad  \lim_{L\to\infty} m_{0,\ga,L} = m_{h_{\rm ext}}
   \end{equation}

 It also follows from the Dobrushin theory
that under $P_{t,\ga,L}$ the spins are weakly correlated:
let $z\ne x$ then
 \begin{eqnarray}
    \label{101.10}
&&|E_{t,\ga,L}[ (\si(x,i)-m_{t,\ga,L})\;|\:  \si(z,i) ]|\le  2\;\sum_{n\ge 0} \sum^*_{y_1,j_1,..,y_n,j_n}
r_{\ga,L}(x,i;y_1,j_1)\cdots \nn\\&& \hskip1cm\cdots r_{\ga,L}(y_{n-1},j_{n-1};y_n,j_n)
r_{\ga,L}(y_n,j_n;z,i)
   \end{eqnarray}
where the $*$sum means that all the pairs
$(y_k,j_k), k=1,..,n$ must differ from $(z,i)$.  Thus
there is   a constant $c$ so that
 \begin{eqnarray}
    \label{101.11}
&&|E_{{t,\ga,L}}[ (\si(x,i)-m_{t,\ga,L})\;|\:  \si(z,i) ]|  \le c \ga
   \end{eqnarray}
and also (after using Chebitchev)
 \begin{equation}
    \label{101.12}
E_{{t,\ga,L}}\Big[ |\sum_y J_{\ga,L}(x,y) (\si(y,i)-m_{t,\ga,L})| \Big]\le c \ga
   \end{equation}

We can also use the Dobrushin technique to estimate the Vaserstein distance
between  $P_{t,\ga,L}$ and $\nu^0_L$.  The key bound is again the
Vaserstein distance between single spin conditional expectations.  We have
\begin{eqnarray}
    \label{101.9.000}
 &&\hskip-.5cm
 E_{t,\ga,L}[ \si(x,i)\;|\: \{\si(y,j)\}]-E_{\nu^0_L}[ \si(x,i)\;|\: \{\si'(y,j)\}]\nn\\&& \hskip.5cm =\tanh \Big\{ t\sum_{y\ne x} J_{\ga,L}(x,y) \si(y,i)+(1-t)m_{h_{\rm ext}}
+\la[\si(x,i+1)+\si(x,i-1)] + h_{\rm ext}\Big\}\nn\\&&\hskip1cm - \tanh \Big\{
+\la[\si'(x,i+1)+\si'(x,i-1)] + h_{\rm ext}+ m_{h_{\rm ext}}\Big\}
   \end{eqnarray}
thus, calling $A:= \cosh^{-2}( h_{\rm ext}-1-2\la)$, we can bound the absolute value of the left hand side of \eqref{101.9.000}  by:
 \begin{eqnarray*}
    \label{101.9.0000}
 && \sum_{j=i\pm 1} r_{\ga,L}(x,i;x,j) |\si(x,j)-\si'(x,j)| + At |
 \sum_y J_{\ga,L}(x,y) \si(y,i) - m_{h_{\rm ext}}|
    \end{eqnarray*}
After adding and subtracting $m_{t,\ga,L}$ to each $\si(y,i)$ and recalling that
$\sum_y  J_{\ga,L}(x,y)=1$,  we use the  Dobrushin analysis to claim that
there exists a joint representation $\mathcal P_{t,\ga,L}$ of $P_{t,\ga,L}$ and $\nu^0_L$ such that
  \begin{eqnarray}
    \label{101.14}
&& \mathcal E_{t,\ga,L}[ |\si(x,i)-  \si'(x,i)|] \le \sum_{j=i\pm 1} r_{\ga,L}(x,i;x,j)
\mathcal E_{t,\ga,L}[ |\si(x,j)-\si'(x,j)|]\nn\\&&\hskip.3cm +At\Big(
\mathcal E_{t,\ga,L} [|\sum_y J_\ga(x,y) (\si(y,i)-m_{t,\ga,L})|] +|m_{t,\ga,L}-m_{h_{\rm ext}}|\Big)
   \end{eqnarray}
Since $\sum_y J_\ga(x,y) (\si(y,i)-m_{t,\ga,L})$ does not depends on $\si'$ we can replace the $\mathcal E_{t,\ga,L}$ expectation by the $ E_{t,\ga,L}$ expectation and after using \eqref{101.12} we get by iteration
\begin{equation}
    \label{101.15}
 \mathcal E_{t,\ga,L}[ |\si(x,i)-  \si'(x,i)|] \le \frac{At}{1-r} \Big(
c\ga +|m_{t,\ga,L}-m_{h_{\rm ext}}|\Big)
   \end{equation}
with $r$  as  in \eqref{101.8}.
Since $|m_{t,\ga,L}-m_{0,\ga,L}|\le  \mathcal E_{t,\ga,L}[| \si(x,i)]-  \si'(x,i)|]$,
\eqref{101.15} yields
\begin{equation*}
    \label{101.15}
 |m_{t,\ga,L}-m_{0,\ga,L}| \le \frac{At}{1-r} \Big(c\ga
 +|m_{t,\ga,L}-m_{0,\ga,L}| +|m_{0,\ga,L}-m_{h_{\rm ext}}|\Big)
   \end{equation*}
By \eqref{101.8}  $\frac{At}{1-r} \le \frac{r}{1-r}   < \frac 13$, so that
\begin{eqnarray}
    \label{101.16}
&& \frac 23 |m_{t,\ga,L}-m_{0,\ga,L}| \le  \frac 13\Big( c\ga+|m_{0,\ga,L}-m_{h_{\rm ext}}|\Big)\nn\\&&
 |m_{t,\ga,L}-m_{h_{\rm ext}}| \le |m_{h_{\rm ext}}-m_{0,\ga,L}|+
 (c\ga+|m_{0,\ga,L}-m_{h_{\rm ext}}|)
   \end{eqnarray}
Thus $m_{t,\ga,L} \to m_{h_{\rm ext}}$ as first
$L\to \infty$ and then $\ga\to 0$.  This holds
for all $t$ and in particular for $t=1$ hence properties (i) and (ii) are proved. Moreover, since $m_\ga\equiv m_{1,\ga}$ converges as $\ga\to 0$ to $m_{h_{\rm ext}}$ the latter is uniquely determined, as a consequence the equation $ h_{\rm ext} +m= f'_\la(m)$ has a unique solution
$m_{h_{\rm ext}}$ which is the limit of $m_\ga$ as $\ga\to 0$.
To prove (iii) we go back to \eqref{102.3.02} and observe that
   \begin{equation*}
H^0_L-
H^{\rm per}_{\ga, h_{\rm ext},L}
= \sum_{(x,i)\in \La} \si(x,i)\Big(\frac 12\sum_{y\ne x}J_\ga(x,y)\si(y,i)
- m_{h_{\rm ext}}\Big)
   \end{equation*}
   Therefore
     \begin{eqnarray*}
&&|E_{t,\ga,L}[H^0_L-
H^{\rm per}_{\ga, h_{\rm ext},L}]-\sum_{(x,i)\in \La}  E_{t,\ga,L}[\si(x,i)]\Big(\frac {m_{t,\ga,L}}2
- m_{h_{\rm ext}}|\Big)|\nn\\&& \hskip1cm\le \sum_{(x,i)\in \La} \frac 12 E_{t,\ga,L}[|
\sum_{y\ne x}J_{\ga,L}(x,y)(\si(y,i)-m_{t,\ga,L})|]
  \le |\La| c\ga
   \end{eqnarray*}
   \eqref{102.3.02} and  \eqref{102.3.03} then yield \eqref{102.3.00}
 because $m_{t,\ga,L} \to m_{h_{\rm ext}}$ as $L\to \infty$ and then $\ga\to 0$.
 This is the same as taking the inf over all $m$ because we have already seen that
$ h_{\rm ext} +m= f'_\la(m)$ has a unique solution.

\vskip1cm

\setcounter{equation}{0}

\section{Proof  of Theorem \ref{Fthm4.1}}
\label{app.B}

%
%
%
%
%
%
%

Following Lebowitz and Penrose we do coarse graining on a scale $\ell$,
$\ell$ the integer part of $\ga^{-1/2}$. Without loss of generality
we restrict $L$ in \eqref{102.2} to be an integer multiple of $\ell$.  We then split each horizontal line in $\La$ into $L/\ell$ consecutive intervals of length $\ell$ and call
$\mathcal I$ the collection of all such intervals in $\La$.  Thus
\[
\mathcal M_\ell = \{ -1, -1+ \frac{2}{\ell},\dots,1- \frac{2}{\ell},  1 \}
\]
is the set of all possible values of the empirical spin magnetization in an interval $I\in \mathcal I$.  We
denote by $\und M$ the set of all functions
$\und m=\{m(x,i), (x,i)\in \La\}$ on $\La$ with values
in $\mathcal M_\ell$ which are constant on each one of the intervals $I$ of $\mathcal I$.  Due to the smoothness assumption on the Kac potential there is $c$ so that for all $\si$, $\ga$ and $L$
 \begin{equation}
    \label{102.4}
\Big| \sum_{(x,i)\in \La}\frac{1}{2}J_{\ga,L}(x,y)\big(\si(x,i)-m(x,i|\si)\big) \Big| \le c\ga^{1/2} \La
   \end{equation}
where, denoting by $I_{x,i}$ the interval in $\mathcal I$ which contains $(x,i)$,
 \begin{equation}
    \label{102.5}
m(x,i|\si)= \frac 1\ell \sum_{y:(y,i)\in I_{x,i}} \si(y,i)
   \end{equation}
Thus $m(x,i|\si)$ does not change when $(x,i)$ varies in an interval of $\mathcal I$
and therefore $\und m=\{m(x,i|\si),(x,i)\in \La\}\in \und M$.
Then the partition function
 \begin{equation}
    \label{102.6}
 Z_{\ga, L}:= \sum_{\und m\in \und M} e^{\frac 12  \sum_{i,x,y} J_{\ga,L}(x,y)
 m(x,i)m(y,i)+
 \sum_{x,i}h_{\rm ext}m(x,i)} \sum_\si \mathbf 1_{\und m(\cdot|\si) = \und m} e^{- H^{{\rm vert}}_L(\si)}
   \end{equation}
has the same asymptotics as $Z_{\ga, h_{\rm ext},L}^{\rm per}$ in the sense that
 \begin{equation}
    \label{102.7}
 \lim_{\ga\to 0}\lim_{L\to \infty} \frac 1{|\La|}\Big|\log Z_{\ga, L}-
 \log Z_{\ga, h_{\rm ext},L}^{\rm per}\Big| =0
   \end{equation}
We next change the vertical interaction $H^{{\rm vert}}_L(\si)$ by replacing
\[
-\la \si(x,n\ell)\si(x,n\ell+1) \to -\la \si(x,n\ell)\si(x,(n-1)\ell+1)
\]
and call $H^{{\rm vert}}_\ell(\si)$ the new vertical energy.  We then split each
vertical column into intervals of length $\ell$, calling $I'$ such intervals
and $\Delta$ the squares $I\times I'$. Let $\Delta=I\times I'$, $m_\Delta$ the restriction of $\und m$ to $\Delta$, so that $m_\Delta(x,i)$, $x\in I, i\in I'$ is only a function of
$i$ with values in $\mathcal M_\ell$.  Recalling the definition \eqref{F4.3} of
$\phi_\ell(m_\Delta)$
we have that $Z_{\ga, L}$ has the same asymptotics as
\begin{equation}
    \label{102.9}
 Z_{\ga, L,\ell}:= \sum_{\und m\in \und M} e^{\frac 12  \sum_{i,x,y} J_{\ga,L}(x,y)
 m(x,i)m(y,i)+
 \sum_{x,i}\{h_{\rm ext}m(x,i)-\phi_\ell(m_{\Delta_{x,i}})\} }
   \end{equation}
where $\Delta_{x,i}$ denotes the square $\Delta$ which contains $(x,i)$.

The cardinality of  $\und M$ is $\dis{\ell^{|\La|/\ell}}$, hence
$Z_{\ga, L,\ell}$ has the same asymptotics as
\begin{equation}
    \label{102.10}
 Z^{\rm max}_{\ga, L,\ell}:= \max_{\und m\in \und M} e^{\frac 12  \sum_{i,x,y} J_{\ga,L}(x,y)
 m(x,i)m(y,i)+
 \sum_{x,i}\{h_{\rm ext}m(x,i)-\phi_\ell(m_{\Delta_{x,i}})\} }
   \end{equation}
Recalling the definition \eqref{F4.4} of $ Z^{\rm max}_{\Delta}$,
we are going to show that
\begin{equation}
    \label{102.12}
 \frac 1{|\La|}\log Z^{\rm max}_{\ga, L,\ell} =  \frac 1{|\Delta|}
 \log  Z^{\rm max}_{\Delta}
   \end{equation}
To prove \eqref{102.12} we write
\[
 m(x,i)m(y,i) = \frac 12 \Big(m(x,i)^2+m(y,i)^2\Big) -
  \frac 12 \Big(m(x,i)-m(y,i)\Big)^2
\]
and use that $\sum_y J_{\ga,L}(x,y)=1$.  In this way the exponent in the right hand side of \eqref{102.10} becomes a sum over all the squares $\Delta$ of terms which depend on $m_\Delta$
plus an interaction given by
\[
-\sum_{i,x,y} J_{\ga,L}(x,y)
  \frac 12 \Big(m(x,i)-m(y,i)\Big)^2
\]
Due to the minus sign the maximizer is obtained when all $m_\Delta$ are equal to each
other and to the maximizer
in  \eqref{F4.4}.  To complete the proof of \eqref{102.12} we still need to prove
the bound on the magnetization:

%
%

\medskip

\begin{prop}
There are $\la_0>0$ and $m_+< 1$ so that for any $\la\le \la_0$ the maximum
in \eqref{102.10} is achieved on configurations $m_\Delta$ such that
for all $(x,i)\in \Delta$,
$|m_\Delta(x,i)| \le m_+$.

\end{prop}

\medskip

\begin{proof}
Given $h>0$ let $S(m)$ be the  entropy defined in \eqref{101.4.00} and let $m_h$ be such that
   \begin{equation}
    \label{102.15}
 -[S'(m_h) + m_h]= h
   \end{equation}
Call $m^*$ the value of $m_h$ at $h^*$, $h^*$ as in \eqref{101.17} and   choose $m_+> m^*$.
Fix any horizontal line $i$ in $\Delta$, take
a magnetization $m_i$ such that $m_i\ge m_+$, it is then sufficient to prove that
for all $\si(x,i+1)+\si(x,i-1) = : h_i(x)$,
  \begin{equation}
    \label{102.16}
e^{-\ell U(m_i)} \sum_{\si} \mathbf 1_{\sum \si(x) = \ell m_i}
e^{\la \sum_x \si(x) h_i(x)}
 \le e^{-\ell U(m^*)} \sum_{\si} \mathbf 1_{\sum \si(x) = \ell m^*}
e^{\la \sum_x \si(x) h_i(x)}
   \end{equation}
where $\dis{U(m) = - \frac{m^2}2- h_{\rm ext} m}$. Since $|h_i| \le 2$, this is implied (for $\ell$ large enough) by
  \begin{equation}
    \label{102.17}
 -  U(m_i)  +S(m_i) + 4\la
 <  - U(m^*)  +S(m^*)
   \end{equation}
Since $m_i>m^*$ and $h_{\rm ext} \le h^*$, \eqref{102.17} is implied by
  \begin{equation}
    \label{102.21}
  \frac{m_i^2}2 +h^* m_i + S(m_i) +4\la <  \frac{(m^*)^2}2 +h^* m^*  +S(m^*)
   \end{equation}
The function $m^2+S(m) + h^*m$ is strictly concave in a neighborhood of $m^*$ where
it reaches its maximum, hence (recalling that $m_i\ge m_+>m^*$
 \begin{equation*}
   \Big( \frac{(m^*)^2}2 +h^* m^*  +S(m^*)\Big)
  -\Big( \frac{m_i^2}2 +h^* m_i + S(m_i)\Big)
   \end{equation*}
is strictly positive and \eqref{102.16} follows for  $\la$ small enough.

\end{proof}

\setcounter{equation}{0}

\section{Cluster expansion}
\label{app.C}

In this appendix we will  study the partition function
$Z^*_{\ell,\und h }$ defined in \eqref{103.1}
using cluster expansion.

\vskip.5cm

\subsection{Reduction to a gas of polymers}
\label{Fsubsec.C.1}
We shall first prove in Proposition \ref{Fprop.5.1} below that  $Z^*_{\ell,\und h }$ can be written as
the partition function of a gas of polymers  $\Ga$.  The definition of polymers and the main notation of this section are given below.

\begin{itemize}

\item $\Ga = (C,S,X)$ denotes a polymer, $C$
its spatial support,  $X$ and $S$ its specifications.
$C$ is a collection of pairs of consecutive points ($L$ and $1$ being consecutive points),
and calling connected
two pairs  if they have a common point, then $C$ is connected. We write
$i\in C$ or sometimes $i\in \Ga$ if $i$ is in one of the pairs of $C$.
Each pair in $C$ is either a $X$-pair or a $S$-pair, $S$ and $X^*$ are the collection of all the  $S$ and respectively all the $X$ pairs.
$X$ is the set of all  points $i$ which belong
to one and only one of the  $X$-pairs.

\item $|C|$ is the number of pairs in $C$, $|S|$ the number of pairs in $S$
and $|X|$ the number of points in $X$. It follows directly from its definition that  $|X|$
is even.

\item $\Ga$ and $\Ga'$ are compatible, $\Ga \sim \Ga'$, if the spatial supports
of  $\Ga$ and $\Ga'$ do not have any point in common.

\item $w(\Ga)$ is the weight of the polymer $\Ga$. If $C$ consists of all the possible pairs
(so that $|C|=\ell$) and
$S=\emptyset$ then we set
\begin{equation}
    \label{1.8.0}
w(\Ga) = \sinh ( \la) ^{\ell} 
   \end{equation}
Otherwise:
\begin{equation}
    \label{1.8}
w(C,S,X) = \sinh ( \la) ^{|C|} \Big( \frac {[\cosh ( \la)-1]}{\sinh ( \la)}\Big)
^{|S|}
\{\prod_{x\in X}  u_{x}\}, \quad   u_{x} = \tanh (h_x)
   \end{equation}

\end{itemize}

Each $X$-pair in $\Ga$ contributes to the weight of $\Ga$ by a factor
$\sinh ( \la)$ while each $S$-pair contributes with a factor
$[\cosh ( \la)-1]$, as it readily follows from \eqref{1.8}.  The dependence of the weight on $h_i$ is through the terms $u_{i}$, $i \in X$.

\medskip

\begin{prop}
\label{Fprop.5.1}
Let $\Ga$ and $w(\Ga)$ be as above, then
\begin{equation}
    \label{1.7}
Z^*_{\ell,\und h }=\sum_{\underline \Ga} \prod_{\Ga\in \und \Ga}w(\Ga)
   \end{equation}
where the sum is over all collections $\und \Ga=\Ga_1,..,\Ga_n$ of mutually compatible polymers.
\end{prop}

\medskip

\begin{proof}
We write
\[
Z^*_{\ell,\und h } = \sum_{\si} \{\prod_i \frac{e^{h\si_i}}{e^{h_i}+e^{-h_i}}\}
\{\prod_i  [e^{\la\si_i \si_{i+1}}-1+1]\}
\]
By expanding the last product we get a sum of terms each one being characterized by
the pairs $(i,i+1)$ with
$e^{\la\si_i \si_{i+1}}-1$.  We fix one of these terms and perform the sum over $\si$.
We call cluster
a maximal connected set of pairs with  $[e^{\la\si_i \si_{i+1}}-1]$, this will be the spatial support of a polymer.
The sum over $\si$
factorizes over the clusters.
After writing
\[
e^{\la\si_i \si_{i+1}}-1 = \sinh(\la)\si_i \si_{i+1} + [\cosh(\la)-1]
\]
we  call $(i,i+1)$ a $X$-pair if it has the term $\sinh(\la)\si_i \si_{i+1}$
and a $S$-pair if it has the term $[\cosh(\la)-1]$.  Notice that if $i$ belongs to two $X$-pairs then
we have a product of two $\si_i$ which is equal to 1.  Thus the sum over the spins in a cluster $C$
becomes a sum over $w(\Ga)$ with the spatial support of $\Ga$ equal to $C$. In this way we
get \eqref{1.7}.

\end{proof}

We shall also consider
the partition function
\begin{equation}
    \label{1.8.1}
Z'_{\ell}=\sum_{\underline \Ga} \prod_{\Ga\in \und \Ga}w_1(\Ga)
   \end{equation}
where $w_1(\Ga)$ 
is obtained from $w(\Ga)$ by putting $u_i\equiv 1$.

\vskip.5cm

\subsection{The K-P condition}

The K-P condition for cluster expansion requires that after introducing
a weight $|\Ga|$ then
for any $\Ga$
  \begin{equation*}
\sum_{\Ga' \not \sim \Ga }|w(\Ga')|e^{ |\Ga'|} \le  |\Ga|
   \end{equation*}

\medskip

\begin{prop}
For $\la$ small enough  we have that
  \begin{equation}
    \label{1.10}
\sum_{\Ga' \not \sim \Ga }|w(\Ga')|e^{ |\Ga'|(1+b)} \le  |\Ga|, \quad b>0
   \end{equation}
with
\begin{equation}
    \label{1.10.1}
 |\Ga| = |C(\Ga)|+1,\quad e^b:= \la^{-5/12}
   \end{equation}
having called $C(\Ga)$ the spatial support of $\Ga$.

\end{prop}

\medskip

\begin{proof}
We are first going to prove that for $\la$ small enough
   \begin{equation}
    \label{1.13}
\sum_{\Ga':  {C'} \ni 1} w_1(\Ga') e^{(1+b)|\Ga'|} \le 1
   \end{equation}
Fix $C$ and consider all $\Ga$ with spatial support $C$, i.e.\ $C(\Ga)=C$,
so that   $|\Ga|= |C|+1=:n$.
Then
   \begin{equation}
    \label{1.12}
\sum_{\Ga : C(\Ga)= C} w_1(\Ga)e^{(1+b)|\Ga|} \le  e^{(1+b)}[\sinh ( \la)e^{(1+b)}] ^{n-1} \Big(
1 + \frac{[\cosh ( \la)-1]}{\sinh ( \la)} \Big)^{n-1}
   \end{equation}
Therefore the left hand side of \eqref{1.13} is bounded by
\[
\sum_{n\ge 2} n e^{(1+b)}[\sinh ( \la)e^{(1+b)}] ^{n-1} \Big(
1 + \frac{[\cosh ( \la)-1]}{\sinh ( \la)} \Big)^{n-1} 
\]
which vanishes when $\la\to 0$, because by
\eqref{1.10.1} $\la e^{2b}$ vanishes as $\la\to 0$.  Hence
\eqref{1.13} holds
for $\la$ small enough.

To prove \eqref{1.10} we first write
  \begin{equation}
    \label{1.10.2}
\sum_{\Ga' \not \sim \Ga }|w(\Ga')|e^{(1+b) |\Ga'|} \le
\sum_{\Ga' \not \sim \Ga }w_1(\Ga')e^{(1+b) |\Ga'|} 
   \end{equation}
and then use \eqref{1.13} to get
\begin{equation*}
\sum_{\Ga' \not \sim \Ga }w_1(\Ga')e^{(1+b) |\Ga'|} \le \sum_{i\in  {C(\Ga)}}
\sum_{\Ga':   {C(\Ga')} \ni i} w_1(\Ga') e^{(1+b)|\Ga'|} \le   |\Ga|
   \end{equation*}

   \end{proof}

   \vskip.5cm

\subsection{The basic theorem of cluster expansion}

The theory of cluster expansion states that if
the K-P condition is satisfied then the log of the partition function can be written as an absolutely convergent series over ``clusters'' of polymers.  To define the clusters it is convenient to regard the space $\{\Ga\}$ of all polymers as a graph where two polymers
are connected if they are incompatible, as defined in Subsection \ref{Fsubsec.C.1}.  Then a cluster is a
connected set in $\{\Ga\}$ whose elements may also have multiplicity larger than 1.  We thus
introduce  functions $I: \{\Ga\} \to \mathbb N$
such that $\{\Ga:I(\Ga)>0\}$ is
a non empty connected
set which is the cluster defined above, $I(\Ga)$ being the multiplicity of appearance of $\Ga$ in the cluster. With such notation the theory says that
\begin{equation}
    \label{1.14}
\log Z^*_{L,\und h }=\sum_{I} W^I, \quad W^I :=a_I\prod_{\Ga}w(\Ga)^{I(\Ga)}
   \end{equation}
   \begin{equation}
    \label{1.14.1}
\log Z'_{L}=\sum_{I} W_1^I, \quad W_1^I :=a_I\prod_{\Ga}w_1(\Ga)^{I(\Ga)}
   \end{equation}
where the sums in \eqref{1.14}--\eqref{1.14.1} are absolutely convergent.
The coefficients $a_I$ are combinatorial (signed) factors, in particular $a_I=1$ if $I$ is supported by a single $\Ga$.
We will not need the explicit expression of the $a_I$ and only use the bound provided by
Theorem \ref{Fthm.5.6} below. We use the notation:
\begin{equation}
    \label{1.15}
    |I|_1 = \sum_{\Ga} I(\Ga),\quad
||I|| = \sum_{\Ga} |\Ga|I(\Ga)
   \end{equation}

\medskip

\begin{thm} [Cluster expansion]
\label{Fthm.5.6}
Let $\la$ be so small that the K-P condition \eqref{1.10} holds. Let $\Ga$ be a polymer and $\mathcal I$ a subset in $\{I\}$ such that
$I(\Ga)\ge 1$ for all $I\in \mathcal I$ ($\mathcal I$ could be the whole $\{I\}$).  Then
\begin{equation}
    \label{1.16}
\sum_{I\in \mathcal I}|W_1^I| e^{||I||} \le w_1(\Ga) e^{(1+b)|\Ga|} \sup_{I\in \mathcal I}
e^{-b||I||}
   \end{equation}

\end{thm}

\medskip

Observe that the absolute convergence of the sum in \eqref{1.14}--\eqref{1.14.1} is implied by
\eqref{1.16} with $\mathcal I=\{I: I(\Ga)\ge 1\}$ as it becomes
\begin{equation}
    \label{1.17}
\sum_{I: I(\Ga)\ge 1}|W_1^I| e^{|I|} \le w_1(\Ga) e^{|\Ga|}
   \end{equation}
because $\inf_{I\in \mathcal I}
e^{-b|I|} = e^{-b|\Ga|}$ as the inf is realized by  $I^*$ which has $I^*(\Ga)=1$ and
$I^*(\Ga')=0$ for all $\Ga'\ne \Ga$. \eqref{1.17}
proves that the sum in \eqref{1.14.1} and hence the sum in \eqref{1.14} are both
absolutely convergent.

   \vskip.5cm

\section{Proof of Theorem \ref{thm103.1}}
In this section we will prove Theorem \ref{thm103.1} as a direct
consequence of
Theorem \ref{Fthm.5.6}.

\medskip

\subsection{Proof of\eqref{103.3} }
We start from \eqref{1.14} and observe that
\begin{equation*}
  W^I :=a_I\prod_{\Ga}w(\Ga)^{I(\Ga)} = \{ a_I\prod_{\Ga}w_1(\Ga)^{I(\Ga)}\}\{
  \prod_{\Ga} (u_{X(\Ga)})^{I(\Ga)}\}
   \end{equation*}
The last factor is equal to $u^{N(\cdot)}$ where $N(\cdot)$ is determined by $I$:
 \begin{equation}
    \label{2.0.0}
N(x) = \sum_{\Ga} I(\Ga) \mathbf 1_{x\in X(\Ga)}
   \end{equation}
hence  \eqref{103.3}. Recalling \eqref{103.7} we observe that
$|N(\cdot)|$ is even
because the cardinality of each $X(\Ga)$ is even.

%
%
%
%
%
%
%
%
%
%

\vskip.5cm

\subsection{The term with $|N(\cdot)|=0$}
The term with $|N(\cdot)|=0$ is a constant $A_{0}$ (i.e.\ it does not depends on $u$)
and it will not play any meaningful role.  It is bounded as follows:

\begin{lem}
There is a constant $c$ (independent of $u$ and $\ell$) such that
\[
|A_{0}| \le c\la^2\ell
\]

\end{lem}

\begin{proof}
By \eqref{1.17}
\begin{eqnarray*}
|A_{0}| &\le& \sum_{i=1}^\ell \sum_{C\ni i}\sum_{I: I(C,C,\emptyset)>0}|W^I|\\
 &\le& \sum_{i=1}^\ell \sum_{C\ni i}[\cosh(\la)-1]^{|C|}e^{|C|+1} \le \ell c \la^2
\end{eqnarray*}

\end{proof}

\medskip

\medskip

\subsection{Proof  of \eqref{103.6}}
We have
 \begin{eqnarray*}
 \alpha_{j-i} u_iu_j &=& \sum_{\Ga=(C,S,X)}
 \mathbf 1_{ X=\{i,j\}} \sum_{I:I(\Ga)= 1; I(\Ga')=0 \;\text{if $ X'\ne \emptyset$ and $\Ga'\ne \Ga$} }
 W^I
  \end{eqnarray*}
Thus by \eqref{1.16}
 \begin{eqnarray*}
 |\alpha_{j-i}| &\le&  \sum_{\Ga=(C,S,X)}
 \mathbf 1_{ X=\{i,j\}} w_1(\Ga) e^{|\Ga|}
 \\ &\le  & \sum_{\Ga=(C,S,X)}
 \mathbf 1_{ X=\{i,j\}} (\sinh(\la))^{|i-j|}(\cosh(\la)-1)^{|C|-|i-j|} e^{|C|+1}
  \end{eqnarray*}
which is bounded by
\[
 |\alpha_{j-i}|
 \le  \sum_{n\ge 0, m\ge 0}
  (\sinh(\la))^{|i-j|}(\cosh(\la)-1)^{n+m} e^{|i-j| +n+m+1}
\]
We have thus proved the second inequality in \eqref{103.6}.

To prove the first one we call $\Ga^*=(C^*,S^*,X^*)$ where $C^*=(i,i+1)$, $S^*=\emptyset$, $X^*=\{i,i+1\}$ and write
 \begin{eqnarray}
     \label{3.000}
  \alpha_{i,i+1} u_iu_{i+1} &=& \sinh(\la)u_iu_{i+1} +
  \sum_{I: |I|_1>1;I(\Ga^*)= 1; I(\Ga')=0 \;\text{if  $\Ga' \ne \Ga^*$ and}\; X'\ne \emptyset}
 W^I \nn\\
  &+&\sum_{\Ga=(C,S,X), \Ga \ne \Ga^*}
 \mathbf 1_{ X=\{i,j\}} \sum_{I:I(\Ga)= 1; I(\Ga')=0 \;\text{if $X'\ne \emptyset$ and $\Ga' \ne \Ga$}}
 W^I
  \end{eqnarray}
If $I$ is  the second  term on the right hand side
then $||I|| \ge 2+2$ so that  this term is bounded by
  \[
 |w(\Ga^*)| e^{|\Ga^*|}e^{-2b} \le \sinh (\la) e^2 e^{-2b}
\]
Proceeding as in the proof of the second  inequality in \eqref{103.6}
 we can bound the last term on the right hand side of \eqref{3.000}  by
  \[
 \le  \sum_{n\ge 0, m\ge 0, n+m>0}
  (\sinh(\la))(\cosh(\la)-1)^{n+m} e^{2 +n+m+1}
\]
  which proves  the first  inequality in \eqref{103.6}.

\medskip

\subsection{Proof  of \eqref{103.8}}
If $I$ determines $N(\cdot)$ then for all $j$
 \begin{equation}
    \label{2.3.0}
N(j) \le \sum_{\Ga: C(\Ga) \ni j} I(\Ga)
   \end{equation}
hence
 \begin{equation}
    \label{2.3}
 R(N(\cdot)) \le \sum_{\Ga: I(\Ga)>0}|\Ga|,\quad |N(\cdot)| =
\sum_j N(j)  \le  \sum_j\sum_{\Ga: C(\Ga) \ni j} |I(\Ga)|
   \end{equation}
Thus
 \begin{equation}
    \label{2.4}
||I|| \ge \|N(\cdot)\|
   \end{equation}
so that the left hand side of  \eqref{103.8} is bounded by:
 \begin{equation}
    \label{2.5}
\sum_{\Ga\ni i}\sum_{I: I(\Ga)>0, ||I|| \ge M
} |W_1^I| \le \sum_{\Ga\ni i} w_1(\Ga)e^{|\Ga|(1+b)} e^{- b M }
   \end{equation}
having used \eqref{1.16}. \eqref{103.8} then follows from
\eqref{1.13}.

\vskip1cm

\setcounter{equation}{0}

\section{A priori bounds}

 We will extensively use the bounds in this section which are corollaries  of Theorem \ref{thm103.1}.

 \medskip

\begin{cor}
There are constants $c_k$, $k\ge 0$, so that for any $i\in \{1,..,\ell\}$, $k\ge 0$ and $ M\ge
4$,
 \begin{equation}
    \label{2.6}
\sum_{N(\cdot):  N(i)>0, |N(\cdot)|>2,\|N(\cdot)\|\ge M} \|N(\cdot)\|^k |A_{N(\cdot)}| \le c_k M^k e^{-bM} = c_k \la^{5/3} e^{-b(M-4)}  
   \end{equation}

\end{cor}
\begin{proof}

It follows from  Theorem \ref{thm103.1}, see \eqref{103.8}.

\end{proof}
%
%

%
%
%


\medskip

\begin{cor}
\label{lemma5.1}

There are constants $c'_k$, $k\ge 1$, so that for any $\ell$ and $i\in [1,\ell]$
      \begin{equation}
    \label{FD.1}
 \sum_{i_1,..,i_{k-1}} | \frac{\partial^{k-1}}{\partial u_{i_1}\cdots \partial
 u_{i_{k-1}}} \frac \partial {\partial u_i} \log Z^*_{\ell,\und h}| \le c'_k \la
   \end{equation}
for any $\la$ as small as required in Theorem \ref{thm103.1}.    Moreover
       \begin{equation}
    \label{FD.2}
 \Psi_i(u)=0 \;\text{ if $|u_i|=1$}
   \end{equation}

\end{cor}

\medskip

\begin{proof}
We write $ \log Z^*_{\ell,\und h}= K_1+K_2$ where $K_1$ is obtained by
restricting the sum on the right hand side
of
\eqref{103.3} to  $|N(\cdot) | \le 2$, $K_2$ is  the sum of the remaining terms.
By
\eqref{103.5}--\eqref{103.6} we easily check that $K_1$ satisfies
the bound in \eqref{FD.1}.  We bound
     \begin{equation*}
 \sum_{i_1,..,i_{k-1}} | \frac{\partial^{k-1}}{\partial u_{i_1}\cdots \partial
 u_{i_{k-1}}} \frac \partial {\partial u_i} K_2|
   \end{equation*}
   by
     \begin{equation*}
\sum_{M>2} \sum_{\|N(\cdot)\|= M, N(i)>0} |N(\cdot)|^k R(N(\cdot))^k
   \end{equation*}
  \eqref{FD.1} then follows from \eqref{2.6}.
  \eqref{FD.2}  follows directly from the definition of $\Psi_i(u)$.

\end{proof}

\medskip

\begin{cor}
\label{cor3}
Recalling \eqref{103.5} and writing $\alpha = \sum_{j>i}\alpha_{j-i}$,
      \begin{equation}
    \label{FD.1.1}
\sum_{i<j}  \alpha_{j-i} u_iu_j = \alpha\sum_i u_i^2 - \frac 12
 \sum_{j>i}\alpha_{j-i}(u_i-u_j)^2
   \end{equation}

\end{cor}

\setcounter{equation}{0}
\vskip1cm

\section{Proof of Theorem  \ref{Fthm5.6} and Theorem
\ref{thm101.4}}
\label{app.D}

We write $\|v\|$ for the sup norm of the vector $v$:
$\|v\|:= \max_{i=1,..,\ell}|v_i|$.

\medskip

\subsection{Proof of Theorem \ref{Fthm5.6}}

Existence.  By \eqref{FD.1} we can use the implicit function theorem to claim existence of a
small enough time $T>0$  such that the equation
 \begin{equation}
    \label{FD.3}
m   = u(t)  + t\Psi (u(t))
   \end{equation}
has a solution $u(t), t\in [0,T]$, such that: $u(0)=m$, $u(t)$ is differentiable and $\|u(t)\|<1$, recall that
$\|m\|<1$.

If $\la$ is small enough
 \eqref{FD.1} with $k=1$ yields
 \begin{equation}
    \label{FD.4}
\max_i \sup _{\|u\|\le 1}\sum_j | \frac {\partial }{\partial u_j}\Psi_i(u)| =: r <1
   \end{equation}
so that the matrix $1+ t \nabla \Psi(u(t))$,
$(\nabla \Psi)_{i,j}  =\frac {\partial }{\partial u_j}\Psi_i $, is invertible for $t \le \min\{T,1\}$ and therefore
for  $t \le \min\{T,1\}$
 \begin{equation}
    \label{FD.5}
  \dot u(t) =f(u(t),t):=-\Big(1 + t \nabla \Psi (u(t))\Big)^{-1}  \Psi(u(t)),\quad u(0)=m
   \end{equation}
By \eqref{FD.4}--\eqref{FD.1} $f(u,t)$ is bounded and differentiable for $t\le 1$
and $\|u\|\le 1$,
thus we can extend $u(t)$ till $\min\{1,\tau\}$ where  $\tau$ is the largest time $ \le 1$
such that $\|u(t)\| \le 1$ for $t\le \tau$.  Thus for $t\le \tau$ \eqref{FD.3} has a solution
$u(t)$ which we claim to satisfy $\|u(t)\|<1$.  To prove the claim we suppose
by contradiction that there is a time $t\le \tau$ and $i$ so that $|u_i(t)|=1$.
By \eqref{FD.3}, $m_i=u_i + t\Psi_i(u)= u_i$ (having used \eqref{FD.2}).  We have thus reached a contradiction  because $\|m\|<1$.  Thus the claim is proved and as a consequence
$\tau =1$ and therefore we have a solution
of \eqref{FD.3} for all $t\le 1$ with.

\noindent
Uniqueness.  Suppose there are two solutions $u$ and $v$.  Then
 \begin{equation*}
 u-v= \Psi(v)-\Psi (u)
   \end{equation*}
   Define $u(s) = su +(1-s)v$, $s\in [0,1]$, then
    \begin{equation*}
 \|u-v\| \le \int_0^1 \|\nabla\Psi (u(s)) (u-v)\|\,  ds
   \end{equation*}
Since  $\|u(s)\| <1$ by \eqref{FD.4} $\|\nabla\Psi (u(s)) (u-v)\| \le r\|u-v\|$,
so that $\|u-v\| \le r \|u-v\|$ and therefore $u=v$.

Boundedness. Calling $u=u(t)$ when $t=1$, by \eqref{FD.3} and  \eqref{FD.1}
 \begin{equation}
    \label{FD.6}
  \|u\| \le \|m\|  + \|\Psi (u )\| \le  \|m\|  + c_1 \la
   \end{equation}
so that if $\|m\| \le m_+$ then for $\la$ small enough  $\|u\| <1$
and therefore there exists $h_+$ such that $\|\und h\| \le h_+$.

\medskip

\vskip.5cm

\subsection {Proof of Theorem \ref{thm101.4}}
\label{Fsubsec.3.3}

Since
\begin{equation}
    \label{FD.7}
 -\phi_\ell(\und m) = \frac 1{\ell^2}\log  \{e^{-\ell \sum_i h_i m_i}
 \sum_{\si\in \{-1,1\}^{\Delta}}
  \mathbf 1_{\und m(\cdot|\si) =  \und m}  e^{-\sum_{x,i}\{ -\la \si (x,i)\si_\Delta(x,i+1)
 -  h_i  \si (x,i)\}}\}
   \end{equation}
we have for free
 \begin{equation}
    \label{FD.8}
      \frac{1}{\ell^2}\log \{e^{-\ell \sum_i h_i m_i}Z_{\ga,\Delta,\und h}\} \ge
-\phi_\ell (\und m)
   \end{equation}
and we are thus left with the proof of a lower bound for $-\phi_\ell (\und m) $.

Call $I_i = \{(x,i): x \le \ell - \ell^a\}$, let $a' \in (\frac 12, a)$ and
\begin{equation}
    \label{FD.9}
   \mathcal B_i=\{ \si(\cdot,i): |\sum_{(x,i)\in I_i} [\si(x,i) -m_i]| \le \ell^{a'}\}
   \end{equation}
Let $\mu$ be the Gibbs probability for the system with vertical interactions and magnetic fields $\und h$. We look for a lower bound for
\begin{equation*}
   \mu\Big[\{ \bigcap_i \mathcal B_i\} \cap\{\und m(\cdot|\si) =  \und m\}\Big]
   \end{equation*}
By the central limit theorem
\begin{equation}
    \label{101.29}
   \mu\Big[ \mathcal B_i^c \Big] \le e^{-b \ell^{2a'-1}},\quad b>0
   \end{equation}
because the spins in $I_i$ are i.i.d. with mean $m_i$.  Moreover
\begin{equation}
    \label{101.30}
   \mu\Big[\{ \und m(\cdot|\si) =  \und m \}\;|\; \{ \bigcap_i \mathcal B_i\}\Big]
   \ge   e^{-4\la \ell^{1+a}} 2^{-\ell^{1+a}}
   \end{equation}
because, given $\dis{\{ \bigcap_i \mathcal B_i\}}$, there is at least one configuration
in the complement of $I_i$ on each horizontal line.
Thus
    \begin{equation*}
   \mu\Big[\{ \bigcap_i \mathcal B_i\} \cap\{\und m(\cdot|\si) =  \und m\}\Big] \ge
   (1-\ell e^{-b \ell^{2a'-1}})e^{-4\la \ell^{1+a}} 2^{-\ell^{1+a}}
   \end{equation*}
hence
    \begin{equation*}
 -\phi_\ell (\und m) \ge  \frac{1}{\ell^2}\log \{e^{-\ell \sum_i h_i m_i}Z_{\ga,\Delta,\und h}\}
 - \frac 1 {\ell^2} \log \{
   (1-\ell e^{-b \ell^{2a'-1}})e^{-4\la \ell^{1+a}} 2^{-\ell^{1+a}}\}
   \end{equation*}
which together with \eqref{FD.8} proves \eqref{101.24}.

\vskip1cm

\setcounter{equation}{0}

\section{Proof of Lemma \ref{Flemma4.4.1}}
\label{FApp.E}


We first write
    \begin{eqnarray}
    \label{2.23}
H^{\rm eff}_{\ell,\und h } &&=\sum_{i=1}^\ell \{ -\frac {u_i^2}2 -(h_{\rm ext}- h_i)u_i
-\log(e^{h_i}+e^{-h_i})\}\nn\\
&&+\sum_{i=1}^\ell \{[h_i-u_i-h_{\rm ext}]\Psi_i -\frac {\Psi_i^2}2  \} - \log Z^*_{\ell,\und h}
+A_\emptyset
   \end{eqnarray}
 We have $\log(e^{h_i}+e^{-h_i}) = h_iu_i +S(u_i)$, the entropy $S(u)$ being defined in \eqref{101.4.00}--\eqref{101.4.00.01}. Thus
 \begin{eqnarray}
    \label{2.25}
H^{\rm eff}_{\ell,\und h }  =\sum_{i=1}^\ell \{T(u_i) - h_{\rm ext} u_i
 + (h_i-u_i- h_{\rm ext})\Psi_i -\frac {\Psi_i^2}2  \} - \log Z^*_{\ell,\und h}
 +A_\emptyset
   \end{eqnarray}

The term with $h_{\rm ext}\Psi_i$ in \eqref{2.25} becomes
   \begin{equation*}
- h_{\rm ext} \sum_i \Phi_i   +\la h_{\rm ext}
\sum_i \Big(u_i^2 u_{i+1}+u_i u_{i+1}^2\Big)-2\la h_{\rm ext}
\sum_i u_i
   \end{equation*}
which can be written as
   \begin{equation}
    \label{Z.2}
- h_{\rm ext}\sum_i  [\Phi_i +2\la u_i] +\la h_{\rm ext}
\sum_i \Big(2u_i^3 -(u_i +u_{i+1})(u_{i+1}-u_i )^2\Big)
   \end{equation}
After an analogous procedure  for the term with $(h_i-u_i)\Psi_i$ we get
\eqref{Z.3}.

\vskip1cm

\setcounter{equation}{0}

\section{Proof of Theorem \ref{Fthm.4.8}}
\label{FApp.F}

We  say that a function $F(\und u)$ is ``sum of one body and gradients squared terms'' if
\[
F(\und u)= \sum_{i=1}^\ell f (u_i) + \sum_{1\le i<j\le \ell}
b _{i,j}(\und u) (u_i-u_j)^2
\]
for some functions $f(u)$ and $b_{i,j}(\und u)$.  Thus \eqref{F4.23}
claims that
$H^{(1)}_{\ell,\und h }$ is ``sum of one body and gradients squared terms''.  We say in short that the  ``gradients squared terms are bounded as desired'' if
\[
 \sum_{1\le i<j\le \ell}
|b _{i,j}(\und u)| (u_i-u_j)^2 \le c \la^{1+\frac{2}{3}} \sum_{  i  }
  (u_i-u_{i+1})^2
\]
Hence \eqref{F4.24} will follow by showing that the  gradients squared terms of
$H^{(1)}_{\ell,\und h }$ are bounded as desired.

We will
examine separately the various terms which contribute to $H^{(1)}$ and prove that each one of them is sum of one body and gradients squared terms and that the latter are bounded as desired.

\medskip

\subsection{The
$\Theta$ term}
\label{FAppsubsec.F.1}

 By \eqref{Z.1.1} 
 \begin{equation*}
 \Theta = \sum_{N(\cdot)\ne 0} A_{N(\cdot)} u^{N(\cdot)}+\frac{\la}2\sum_{i=1}^\ell(u_{i+1}-u_i)^2
  \end{equation*}
Call $\Theta^{(2)}$ the above expression when we restrict the sum to
$N(\cdot): |N(\cdot)|=2$
and  call $\Theta^{(>2)}=\Theta-\Theta^{(2)}$. Thus $\Theta^{(>2)}$ is equal to
the sum of $A_{N(\cdot)}$ over  $N(\cdot): |N(\cdot)|>2$, i.e.\ $|N(\cdot)|\ge 4$,  recall in fact from Theorem \ref{thm103.1} that $A_{N(\cdot)}=0$ if ${N(\cdot)}$ is odd.  We start from
$\Theta^{(2)}$ which,
recalling \eqref{FD.1.1}, is equal to
\begin{equation}
    \label{FApp.1}
  \Theta^{(2)} = \alpha \sum_i u_i^2 -\frac 12\sum_i(\alpha_{1} - \la )(u_{i+1}-u_i)^2 
 -\frac 12 \sum_{i<j, j-i>1} \alpha_{j-i} (u_{j}-u_i)^2 
  \end{equation}
Thus $\Theta^{(2)}$ is sum of one body and gradients squared terms.  To prove that the latter
are bounded as desired we write
\begin{equation}
    \label{FApp.1.1}
 (u_{j}-u_i)^2 \le (j-i) \sum_{k=i}^{j-1}(u_{k+1}-u_k)^2
  \end{equation}
and call $n=k-i\ge 0$, $m= j-k\ge 1$.  We then
use \eqref{103.6} to bound the sum of the terms with the gradients by
\begin{equation}
    \label{FApp.2}
\sum_{k} (u_{k+1}-u_k)^2\Big\{c \la e^{-2b} + \sum_{n\ge 0,m\ge 1, m+n>1} (m+n)c \la^{m+n}e^{m+n}\Big\}
  \end{equation}
which is the desired bound because $\frac 23\le \frac 56$.

We rewrite $\Theta^{(>2)}$ using \eqref{F5.1} for each one of the factors $u^{N(\cdot)}$.
Thus given $N(\cdot)$ we call $i_1<i_2<\cdots <i_k$ the sites where $N(\cdot)>0$
and call $\und n=(N(i_1),..,N(i_k))$.  We then apply \eqref{F5.1} with $u_1 = u_{i_1}, \dots, u_k = u_{i_k}$ so that  $p_i$ and $d_{i,j}$ in  \eqref{F5.1}  become functions of $\und u$ and $N(\cdot)$.  We then get
\begin{equation}
    \label{FApp.3}
  \Theta^{(>2)} = \sum_{N(\cdot):|N(\cdot)|\ge 4} A_{N(\cdot)}\{  \sum_{i:N(i)>0} p_i u_i^{|N(\cdot)|}
  +\sum_{j>i :  N(j)>0,N(i)>0} d_{i,j}  (u_i-u_j)^2\}
  \end{equation}
 which  is sum of one body and gradients squared terms.  To get the desired   bound on the
 latter  we use  \eqref{FApp.1.1} and \eqref{F5.2} to get
 \begin{equation*}
   \sum_{k}(u_k-u_{k+1})^2 \sum_{i,j: j>k\ge i}\{ (j-i)\sum_{N(\cdot):|N(\cdot)|\ge 4,
   N(i)>0, N(j)>0}c |N(\cdot)|^3|A_{N(\cdot)}| \}
  \end{equation*}
Since both $  N(i)>0$, $N(j)>0$ then $j-i\le R(N(\cdot))$ and given   $R(N(\cdot))\ge k-i$
there are at most $R(N(\cdot))$ possible values of $j$.  Therefore
the above expression is bounded by
  \begin{equation*}
   \sum_{k}(u_k-u_{k+1})^2  \sum_{i\le k} \sum_{N(\cdot):|N(\cdot)|\ge 4,
   N(i)>0, R(N(\cdot)) \ge k-i} \|N(\cdot)\|^5|A_{N(\cdot)}| \} 
  \end{equation*}
We upper bound the above if we extend the sum over $N(\cdot)$ such that
\[
|N(\cdot)|\ge 4,
   N(i)>0, \|N(\cdot)\| \ge \ga_{k-i},\quad \ga_{k-i}:=\max\{4,k-i\}
\]
We then apply \eqref{2.6} with $k=5$ to get
  \begin{equation*}
   \sum_{k}(u_k-u_{k+1})^2  \sum_{i\le k}  c_5 \ga_{k-i}^{5} e^{-b\ga_{k-i}}
   = e^{-4b} \sum_{k}(u_k-u_{k+1})^2 \{ \sum_{i\le k}  c_5 \ga_{k-i}^{5} e^{-b(\ga_{k-i}-4)}\}
     \end{equation*}
The curly bracket is bounded by
\begin{equation*}
4^55+\sum_{n\ge 1}(n+4)^5 e^{-bn } \le c 
  \end{equation*}
 Thus also $\Theta^{(>2)}$ is bounded  as desired. 

 \medskip

\subsection{The
term $h_{\rm ext} \sum_i \Phi_i$}
By \eqref{Z.1}
  \begin{equation}
  \label{FApp.3.00}
\Phi_i = (1-u_i^2)\Big( (\alpha_{1}-\la) (u_{i+1}+u_{i-1})
+ \sum_{j>i+1} \alpha_{j-i} u_j
+ \sum_{N(\cdot): |N(\cdot)| \ge 4} N(i)A_{N(\cdot)} u^{N(\cdot)-e_i}\Big)
\end{equation}
where $e_i(j)=0$ if $j\ne i$ and $=1$ if $j=i$.

Call $g_i:=(1-u_i^2)  (\alpha_{1}-\la)$ then the first term contributes to $\sum_i
\Phi_i$ by
  \begin{eqnarray*}
\sum_i \Big( 2g_i u_{i} -(g_i-g_{i+1})(u_i-u_{i+1})\Big) =
\sum_i  2g_i u_{i} + (\alpha_{1}-\la)\sum_i (u_i+u_{i+1}) (u_i-u_{i+1})^2
\end{eqnarray*}
which is  sum of one body and gradients squared terms.
%
By \eqref{103.6} the coefficients of the gradients squared are bounded
in absolute value by  $2c \la e^{-2b}$ which is the desired bound because $\frac 23\le \frac 56$.

By an analogous argument and writing   $g'_i:=(1-u_i^2)$, the contribution of the second term in \eqref{FApp.3.00}
is
  \begin{eqnarray*}
\sum_{i<j} \alpha_{j-i} \Big( 2g'_i u_{i} -(g'_i-g'_{j})(u_i-u_{j})\Big)
= \sum_{i<j} \alpha_{j-i} \Big( 2g'_i u_{i} +(u_i+u_{j}) (u_i-u_{j})^2\Big)
\end{eqnarray*}
which is  sum of one body and gradients squared terms.  We bound the latter using
\eqref{FApp.1.1} and the second inequality in \eqref{103.6} to get
\[
\sum_{k }(u_{k+1}-u_k)^2 \{\sum_{i\le k <j, j-i>2}2 c( e\la)^k\}
%
\]
which is the desired bound because the curly bracket is bounded by $c' \la^2$.%

To write the contribution to $\sum_i \Phi_i$ of the last term in \eqref{FApp.3.00} we introduce the following notation.  Given $N(\cdot): N(i)>0$  we call $N'(\cdot)= N(\cdot)-e_i$
and $N''(\cdot)=N(\cdot)+e_i$.  Let then
 $i_1<i_2<\cdots <i_k$ the sites $j$ where $N'(j) >0$,
 $\und n=(N'(i_1),..,N'(i_k))$ and denote by $p^{-}_j$, $d^-_{j,j'}$ the
corresponding coefficients in \eqref{F5.1}. Similarly let
$i'_1<i'_2<\cdots <i'_k$ the sites $j$ where $N''(j) >0$,
$\und n=(N''(i_1),..,N''(i_k))$ and denote by $p^{+}_j$, $d^+_{j,j'}$ the
 corresponding coefficients in \eqref{F5.1}. Then  the contribution to $\sum_i \Phi_i$ of the last term in \eqref{FApp.3.00} can be written as
\begin{eqnarray}
    \label{FApp.3.000}
 && \sum_{N(\cdot):|N(\cdot)|\ge 4} A_{N(\cdot)} \sum_{i:N(i)>0} N(i) \Big( \sum_{j:N'(j)>0}[p^-_j u_j^{|N(\cdot)|-1} -\sum_{j:N''(j)>0} p^+_ju_j^{|N(\cdot)|+1}]\nn\\&&
  +\sum_{j<j': N'(j)>0,N'(j')>0 }d^-_{j,j'} (u_j-u_{j'})^2-
  \sum_{j<j': N''(j)>0,N''(j')>0 }d^+_{j,j'} (u_j-u_{j'})^2\}
  \end{eqnarray}
which is sum of one body and gradients squared terms.  To bound the latter
we examine the terms with $d^-$, those with $d^+$ are analogous and their analysis is omitted.
For the $d^-$ terms we get the bound:
\begin{eqnarray*}
 && \sum_{N(\cdot):|N(\cdot)|\ge 4} |A_{N(\cdot)}| \sum_{i:N(i)>0} N(i)
 \sum_{j<j': N'(j)>0,N'(j')>0 } c|N(\cdot)|^3 (u_j-u_{j'})^2 \\&&
 \le \sum_{N(\cdot):|N(\cdot)|\ge 4} |A_{N(\cdot)}|
 \sum_{j<j': N(j)>0,N(j')>0 } c|N(\cdot)|^4 (u_j-u_{j'})^2
  \end{eqnarray*}
which has an analogous structure as the gradient term in \eqref{FApp.3}. Its
analysis   is similar
and thus omitted. We have thus
proved that   $h_{\rm ext} \sum_i \Phi_i$ has the desired structure.

 \medskip

\subsection{The term
$\sum_i\Psi_i^2$}
We introduce the following notation: given $i, N(\cdot),N'(\cdot),\si,\si'$,
$\si\in \{-1,1\}$, $\si'\in \{-1,1\}$, $N(i)>0$, $N'(i)>0$, we call
\[
\bar N(\cdot) = N(\cdot)+N'(\cdot),\quad K\equiv K_{i,\bar N(\cdot),\si,\si'}:=\bar N(\cdot)+(\si +\si')e_i
\]
Then   $\sum_i\Psi_i^2$ is equal to
\begin{eqnarray}
    \label{FApp.3?}
 &&\sum_i  \sum_{N(\cdot),N'(\cdot),\si,\si'} N(i)N'(i)A_{N(\cdot)}
 A_{N'(\cdot)}(-1)^{\frac{\si+\si'}2 +1}  \Big( \sum_{j:K(j)>0} p_j(K)u_j^{|K|}\nn\\&&
  +\sum_{j<j':  K(j)>0,   K(j')>0} d_{j,j'}(K) (u_{j'}-u_j)^2\Big)
  \end{eqnarray}
which is sum of one body and gradient squared terms. Let
\begin{eqnarray*}
 &&C_{j,j'}:= \sum_i  \sum_{N(\cdot),N'(\cdot),\si,\si'} N(i)N'(i)|A_{N(\cdot)}|
 |A_{N'(\cdot)}| \sum_{j<j':  K(j)>0,   K(j')>0} |d_{j,j'}(K)| 
  \end{eqnarray*}
then  the gradient squared terms are bounded by $\sum_{j<j'}C_{j,j'}(u_{j'}-u_j)^2$.  We have
 \begin{equation*}
 C_{j,j'} \le 4 \sum_i  \sum_{N(\cdot),N'(\cdot)} N(i)N'(i)|A_{N(\cdot)}||
 A_{N'(\cdot)}|  
  \sum_{j<j':  \bar N(j)>0, \bar N(j')>0} c(|N(\cdot)|+|N'(\cdot)|+2)^3
    \end{equation*}
because 4 is the cardinality of $(\si,\si')$.  Moreover
 \begin{equation*}
C_{j,j'} \le  4c \sum_i  \sum_{N(\cdot),N'(\cdot):\bar N(j)>0, \bar N(j')>0, N(i)>0,N'(i)>0} |A_{N(\cdot)}||
 A_{N'(\cdot)}|   (2|N(\cdot)|)^4(2|N'(\cdot)|)^4
    \end{equation*}
By the symmetry between $N(\cdot)$ and $N'(\cdot)$ we get
with an extra factor 2:
\begin{equation*}
C_{j,j'} \le  8c 4^4\sum_i  \sum_{N(\cdot),N'(\cdot): N(j)>0, \bar N(j')>0, N(i)>0,N'(i)>0} |A_{N(\cdot)}||
 A_{N'(\cdot)}|  |N(\cdot)|^4 N'(\cdot)|^4
    \end{equation*}
Moreover either $R(N(\cdot)) \ge (j'-j)/2$, or $R(N'(\cdot)) \ge (j'-j)/2$ or both, hence
  \begin{eqnarray*}
 && C_{j,j'} \le  8c 4^4  \Big(\sum_{N(\cdot): N(j)>0, R(N(\cdot)) \ge \frac{j'-j}2}
 |A_{N(\cdot)}| |N(\cdot)|^4 \sum_{i:N(i)>0} \sum_{N'(\cdot): N'(i)>0} N'(\cdot)|^4
 \\&&\hskip1cm+\sum_{N(\cdot): N(j)>0}
 |A_{N(\cdot)}| |N(\cdot)|^4 \sum_{i:N(i)>0} \sum_{N'(\cdot): N'(i)>0, R(N'(\cdot)) \ge \frac{j'-j}2} N'(\cdot)|^4\Big)
  \end{eqnarray*}
  By  \eqref{2.6}
    \begin{eqnarray*}
 && C_{j,j'} \le  8c 4^4  \Big(\sum_{N(\cdot): N(j)>0, R(N(\cdot)) \ge \frac{j'-j}2}
 |A_{N(\cdot)}| |N(\cdot)|^4 |N(\cdot)|c_4 e^{-2b  }
 \\&&
 \hskip1cm+\sum_{N(\cdot): N(j)>0}
 |A_{N(\cdot)}| |N(\cdot)|^4 |N(\cdot)| c_4 e^{-b\max\{2,  \frac{j'-j}2\}}\Big)
  \end{eqnarray*}
  Using again  \eqref{2.6}
      \begin{eqnarray*}
 && C_{j,j'} \le  8c 4^4  2c_4 e^{-2b  }c_5 e^{-b\max\{2,  \frac{j'-j}2\}}
 =: c'e^{-2b  }e^{-b\max\{2,  \frac{j'-j}2\}}
  \end{eqnarray*}
Hence
\[
\sum_{j<j'}C_{j,j'} (u_{j'}-u_j)^2 \le
\sum_{k}(u_{k+1}-u_k)^2 \sum_{j,j':j\le k < j'}(j'-j) c'e^{-2b  }e^{-b\max\{2,  \frac{j'-j}2\}}
\]
The last sum is bounded proportionally to $e^{-4b}$ (details are omitted)
which gives the desired bound. 

 \medskip

\subsection{The
term $  \sum_i \xi_i \Phi_i$}
\label{FAppsubsec.F.3}
Recalling \eqref{Z.3.01} and  \eqref{Z.1}  the contribution to $H^{(1)}_{\ell,\und h }$ due to   $\sum_i \xi_i \Phi_i$ is
\begin{eqnarray}
   \label{FAppsubsec.F.3.1}
  \sum_{i=1}^\ell   (h_i-u_i)(1-u_i^2)\{\sum_{j>i+1} \alpha_{j-i}  u_j
+ \sum_{N(\cdot): N(i)>0,|N(\cdot)| \ge 4}N(i) A_{N(\cdot)} u^{N^{(i)}(\cdot)}  \}
   \end{eqnarray}
 We have
 \begin{equation}
 \label{FAppsubsec.F.3.2}
 (h-u)(1-u^2) = \frac{u^3}3-2 \sum_{k=2}^{\infty}  \frac 1{4k^2-1}  u ^{2k+1}
 =: \sum_{k=1}^{\infty} \kappa_k u ^{2k+1}
 \end{equation}
with $|\kappa_k| <1$;  since $|u| \le u_+ < 1$ the series converges exponentially.
We start from the terms with $\alpha_{j-i}$:
\[
 \sum_{i=1}^\ell \sum_{j>i+1} \alpha_{j-i} \sum_{k\ge 1} \kappa_k u_i^{2k+1} u_j
 =  \sum_{i=1}^\ell \sum_{j>i+1} \alpha_{j-i} \sum_{k\ge 1} \kappa_k \{(p_iu_i^{2k+2}+p_j u_j^{2k+2})+   d  (u_i-u_j)^2\}
\]
where $(p_i,p_j)$ is the probability vector introduced in Theorem \ref{Fthm.5.9}
and $d$ the corresponding coefficient.  They depend on the pair $(2k+1,1)$ and
$|d| \le c k^{6}u_+^{2k}$.  This is sum of one body and squared gradients terms
and we are left with   bounding   the latter.  We have the bound
\[
   \sum_{i=1}^\ell \sum_{j>i+1} |\alpha_{j-i}| \sum_{k\ge 1}  ck^6  u_+^{2k}  (u_i-u_j)^2
   \le  \sum_{i=1}^\ell \sum_{j>i+1} |\alpha_{j-i}|    c' (u_i-u_j)^2
\]
which satisfies the desired bound as proved in Subsection \ref{FAppsubsec.F.1}.

We next study the last term  on the right hand side of \eqref{FAppsubsec.F.3.1}.  Proceeding as before we check that it is sum of one body and gradients squared terms and next prove that
the gradients are bounded as desired. We first bound them by
\[
 \sum_{i}\sum_{j<j'}  \sum_{N(\cdot): N(i)>0,N(j)>0, N(j')>0, |N(\cdot)|\ge 4}N(i)
 |A_{N(\cdot)}| \sum_{k\ge 1} c(2k+|N(\cdot)|)^3 u_+^{2k}(u_{j'}-u_j)^2
 \]
We have $(2k+|N(\cdot)|)^3 \le (2k)^3 |N(\cdot)|^3$ so that we get the bound
\[
 \sum_{i}\sum_{j<j'}  \sum_{N(\cdot): N(i)>0,N(j)>0, N(j')>0, |N(\cdot)|\ge 4}N(i)
 |A_{N(\cdot)}| c' |N(\cdot)|^3 (u_{j'}-u_j)^2
 \]
 with
\[
c':= \sum_{k\ge 1}(2k)^3 u_+^{2k}
\]
We can perform the sum over $i$ to get
\[
 \sum_{j<j'}  \sum_{N(\cdot):  N(j)>0, N(j')>0, |N(\cdot)|\ge 4}
 |A_{N(\cdot)}| c' |N(\cdot)|^4 (u_{j'}-u_j)^2
 \]

We are thus reduced to the case considered  in Subsection \ref{FAppsubsec.F.1}, we omit the details.

\vskip2cm

\vskip1cm

\setcounter{equation}{0}

\section{Proof of Proposition \ref{Fprop.4.8.1}}
\label{FApp.G}

Recalling that $\xi(u):=(h(u)-u)(1-u^2)$, we have, supposing $u'>u'$,
  \begin{equation}
   \label{2.25.3}
\xi(u')-\xi(u)= \int_{u }^{u'} \frac{d\xi}{du} du \le a(u_i-u_j),
\end{equation}
with $\dis{a = \max_{|u| < 1}\frac{d\xi}{du}}$. Thus $\theta_i(\und u) \le a$ and by \eqref{FAppsubsec.F.3.2}
\begin{eqnarray*}
a
&=& \max_{|u|<1}\left(u^2-2 u \sum_{k=1}^{\infty}\frac{u^{2k+1}}{2k+1} \right)
 < \max_{|u|<1} \left(u^2-\frac{2}{3}u^4\right)=\frac{3}{8}
\end{eqnarray*}
having retained only the term with $k=1$.

\vskip1cm

\setcounter{equation}{0}

\section{Proof of Theorem \ref{Fthm.4.8.2}}
\label{FApp.H}


We shall use in the proof that in $H^{\rm eff}_{\ell,\und h }$ all terms but $\left(T(u) - h_{\rm ext} u\right)$, cf. \eqref{2.25}, are proportional to $\lambda$.\\
Calling $\tilde u$ the minimizer of $\left(T(u) - h_{\rm ext} u\right)$ :
\begin{itemize}

\item It will follow from Lemma \ref{lemma6.3} that the minimizer
$\und u^*$ of $H^{\rm eff}_{\ell,\und h }$ has components $u^*_i$ such that $|u^*_i-\tilde u|
< \la^{1/4}$ (for all  $\la$ small enough), and that the minimizer $v$ of
$f(u)$, $f(u)$ the one body term defined in \eqref{F4.23}, is such that $|v-\tilde u|< \la^{1/4}$;

\item Since the gradient of $H^{\rm eff}_{\ell,\und h }$ vanishes at $\und v=(v_i=v,\;i=1,..,\ell)$, cf. \eqref{F4.23}, $\und v$ is a critical point of $H^{\rm eff}_{\ell,\und h }$;

\item $T(u)$ is a convex function and its second derivative $T''(u)$ is a
  strictly increasing, positive function of $u \in (0,1)$ which diverges as
  $u\to 1$, as it follows from \eqref{2.24}.  Then the matrix $\frac {\partial^2}{\partial u_i\partial u_j}H^{\rm eff}_{\ell,\und h }$ is positive definite in the ball
   $\und u: |u_i-\tilde u|
< \la^{1/4}$, cf. Proposition \ref{prop6.5}.

\end{itemize}

As a consequence, the minimizer  of $H^{\rm eff}_{\ell,\und h }$ in the ball coincides with
$\und v$ and since   $\und u^*$ is in the ball it coincides with
$\und v$, thus proving that all the components of $\und u^*$ are equal to each other.
We are thus left with the proof of Lemma \ref{lemma6.3} and  Proposition \ref{prop6.5}.
We need a preliminary lemma.

\medskip

\begin{lem}
\label{lemma6.2}

For any $h_{\rm ext} \in [h_0,h^*]$ there is a unique $\tilde u $ such that
    \begin{equation}
    \label{2.24.1}
 \frac {d}{du}\{ T(  u) - h_{\rm ext}   u\}\Big|_{u=\tilde u} = 0
   \end{equation}
and there is $c_{h_0}>0$ so that
   \begin{equation}
    \label{2.24.2}
\inf_{h_{\rm ext} \in [h_0,h^*]} \frac {d^2}{du^2}  T(u)\Big|_{u=\tilde u}
\ge c_{h_0}
   \end{equation}

\end{lem}

\begin{proof}
The proof follows from the fact that the second derivative
of $T(u)$ is positive away from 0 and in $(0,1)$ increases to $\infty$ as $u\to 1$.

\end{proof}

\medskip

Fix all $u_j, j\ne i$ and call
$F(u_i)$ the energy $H^{\rm eff}_{\ell,\und h }(\und u)$ as a function of $u_i$.  Then

\medskip

\begin{lem}
\label{lemma6.3}
There is  $c'_{h_0}>0$ so that for all $\la$ small enough the following holds.
Let $h_{\rm ext} \in [h_0,h^*]$ and $\tilde u$ as in Lemma \ref{lemma6.2} then
   \begin{equation}
    \label{2.24.3}
\inf_{u_i: |u_i-\tilde u| \ge \la^{1/4}}F(u_i) \ge  F(\tilde u) +   c'_{h_0}
\la^{1/2}
   \end{equation}

\end{lem}

\begin{proof}
By \eqref{2.24.2}
\[
\inf_{u_i: |u_i-\tilde u| \ge \la^{1/4}} |
\{ T(u) - h_{\rm ext} u\} - \{ T(\tilde u) - h_{\rm ext} \tilde u \}| \ge \frac {c_{h_0}}2 \la^{1/2}
\]
We are going to show that the variation of all the other terms in \eqref{2.25} are bounded proportionally to $\la$ and this will then complete the proof of the lemma.  We have
\[
|(h_i-u_i) (1-u_i^2)| \le c, \quad (1-u_i^2)^{-1}|\Psi_i| \le c \la
\]
(the first inequality by \eqref{FAppsubsec.F.3.2}, the last inequality by \eqref{FD.1}).  \\
Call $G(u_i)$ the value of $\log Z^*_{\ell,\und h}$ when $\tanh (h_i)= u_i$ and the other $h_j$ are fixed, then
\[
|G(u_i) - G(u'_i)| =|\sum_{N(\cdot): N(i)>0}
A_{N(\cdot)} u ^{N^{(i)}(\cdot)}(u_i-u'_i)| \le c \la |u_i-u'_i|
\]

where, to derive the last inequality, we have used Theorem \ref{thm103.1}.

\end{proof}

\medskip

As a corollary of the above lemmas

\medskip

\begin{lem}
\label{lemma6.4}
For $\la$ small enough the inf of $H^{\rm eff}_{\ell,\und h }$ is achieved in the ball
$\und u: \max \{ |u_i-\tilde u | \le \la^{1/4}, i=1,..,\ell\}$.

\end{lem}

\medskip

\begin{prop}
\label{prop6.5}
For $\la$ small enough the matrix  $\frac{\partial^2}{\partial u_i\partial u_j} H^{\rm eff}_{\ell,\und h }$ is strictly positive in the ball
$\und u: \max \{ |u_i-u_{h_{\rm ext} }| \le \la^{1/4}, i=1,..,\ell\}$.

\end{prop}

\begin{proof}
From Lemma \ref{lemma6.2} and Corollary \ref{lemma5.1} one obtains
$$
\frac{\partial^2}{\partial u_i^2} H^{\rm eff}_{\ell,\und h }\geq c_{h_0}-\lambda c_1, \text{for} \quad i=1,2,...,L
$$
For any $i$,
$$
\sum_{j\ne i}|\frac{\partial^2}{\partial u_i \partial u_j} H^{\rm eff}_{\ell,\und h }|\leq c_2\lambda
$$
from \eqref{103.8} and Corollary \ref{lemma5.1}.
\end{proof}

\end{document}